\documentclass{article}
\usepackage{enumitem}

\usepackage{tcolorbox}
\tcbuselibrary{skins}
\tcbuselibrary{breakable}
\usepackage{float}

\usepackage{array}
\newcolumntype{L}[1]{>{\raggedright\arraybackslash}p{#1}}

\usepackage{adjustbox}
\usepackage{booktabs,tabularx,makecell,array}
\usepackage{placeins}
\newcolumntype{L}{>{\raggedright\arraybackslash}X}
\newcolumntype{C}{>{\centering\arraybackslash}m{1.6cm}}
\newcolumntype{N}{>{\centering\arraybackslash}m{1.8cm}}
\usepackage[final]{microtype}      % gentle line expansion/protrusion
\usepackage{listings}
\tcbset{width=\linewidth, left=1mm, right=1mm, boxsep=1mm}

\newcommand{\rhoa}{\rho_{\mathrm{a}}}
\newcommand{\rhog}{\rho_{\mathrm{g}}}
\usepackage[utf8]{inputenc}
\usepackage[T1]{fontenc}

\usepackage{alphalph}
\renewcommand*{\Alph}[1]{\alphalph{\value{#1}}}

\usepackage{pgfplots}
\pgfplotsset{compat=1.18}

\usepackage{tabularx}
\newcolumntype{Y}{>{\raggedright\arraybackslash}X} % flexible, wraps nicely

\usepackage{amsmath,amssymb,amsfonts,amsthm}
\usepackage{textgreek} 
\usepackage{natbib}
\usepackage{algorithmicx}
\usepackage{algpseudocode}
\newtheorem{lemma}{Lemma}
\newtheorem{corollary}{Corollary}
\newtheorem{remark}{Remark}
\newtheorem{theorem}{Theorem}
\newtheorem{proposition}{Proposition}
\newtheorem{conjecture}{Conjecture}
\usepackage{hyperref}
\usepackage{braket}
\usepackage{fullpage}
\usepackage{algorithm}

\usepackage{tikz}
\usetikzlibrary{positioning,calc,arrows.meta} % needed for right=of, |- and Stealth tips
\usepackage{amsmath}                           % for \tfrac etc.

\usepackage{tabularx} % add to preamble if not already there

\usepackage{tikz}
\usetikzlibrary{arrows.meta,positioning,calc,shapes.geometric}

\usepackage{subcaption}
\usetikzlibrary{arrows.meta,positioning,calc}
\tikzset{
  loopbox/.style = {draw, rounded corners=2pt, minimum width=28mm, minimum height=12mm, fill=black!3},
  vtx/.style     = {circle, draw, minimum size=7.5mm, inner sep=0pt},
  phase/.style   = {circle, draw, inner sep=0.5pt, minimum size=3.6mm},
  unitary/.style = {draw, rounded corners=2pt, minimum width=26mm, minimum height=11mm, fill=black!5},
  det/.style     = {draw, minimum width=9mm, minimum height=6mm},
  edge/.style    = {-Latex, line width=0.8pt, line cap=round},
  lab/.style     = {font=\small}
}

\usepackage{subcaption}
\usetikzlibrary{arrows.meta,positioning,calc}
\tikzset{
  loopbox/.style = {draw, rounded corners=2pt, minimum width=28mm, minimum height=12mm, fill=black!3},
  vtx/.style     = {circle, draw, minimum size=7.5mm, inner sep=0pt},
  phase/.style   = {circle, draw, inner sep=0.5pt, minimum size=3.6mm},
  unitary/.style = {draw, rounded corners=2pt, minimum width=26mm, minimum height=11mm, fill=black!5},
  det/.style     = {draw, minimum width=9mm, minimum height=6mm},
  edge/.style    = {-Latex, line width=0.8pt, line cap=round},
  lab/.style     = {font=\small}
}

\usetikzlibrary{arrows.meta,positioning,calc}
\tikzset{
  vtx/.style={circle, draw, minimum size=7.5mm, inner sep=0pt},
  ctx/.style={draw, rounded corners=2pt, minimum width=18mm, minimum height=8mm, fill=black!5},
  edge/.style={-Latex, line width=0.8pt, line cap=round}
}

\usepackage{subcaption} % for side-by-side subfigures
\usetikzlibrary{arrows.meta,positioning,calc}
\tikzset{
  loopbox/.style = {draw, rounded corners=2pt, minimum width=28mm, minimum height=12mm, fill=black!3},
  vtx/.style     = {circle, draw, minimum size=7.5mm, inner sep=0pt},
  phase/.style   = {circle, draw, inner sep=0.5pt, minimum size=3.6mm},
  unitary/.style = {draw, rounded corners=2pt, minimum width=34mm, minimum height=12mm, fill=black!5},
  det/.style     = {draw, minimum width=9mm, minimum height=6mm},
  edge/.style    = {-Latex, line width=0.8pt, line cap=round},
  lab/.style     = {font=\small}
}

\usetikzlibrary{calc,positioning}
\tikzset{
  vertex/.style={circle, draw, minimum size=8mm, inner sep=0pt},
  edge/.style={line width=0.8pt, line cap=round, line join=round},
  labelfont/.style={font=\small}
}

\title{Certified Pruning from Counterfactual Consistency: Exact Certificates and Structured SAT Families}

\author{Maximilian R.~P.~von~Liechtenstein\\
\small Independent Researcher}

\date{\today}

\usepackage{amsmath,amssymb,amsthm}
\usepackage{listings}
\usepackage{bbm}
\usepackage{tikz}
\usepackage{pgfplots}
\pgfplotsset{compat=1.18}
\usetikzlibrary{pgfplots.colormaps}
\pgfplotsset{colormap={redyellow}{rgb=(1,0,0) rgb=(1,1,0)}}

\begin{document}
\maketitle

\begin{abstract}
We introduce a certified pruning framework that consolidates the
principles of counterfactual consistency and their networked
extensions into a single operational model, with consequences for both
quantum foundations and cryptographic hardness.  First, we formalize
$\varepsilon$-counterfactual instrumentation and $\varepsilon$-stability,
capturing noisy but testable constraints in laboratory contextuality
experiments.  Second, we extend these constraints to networks of
contexts, yielding contextuality-type inequalities that apply globally
across a CNF-SAT instance.  Third, we implement a propagate-and-prune
solver in which every learned clause is certified by a dual Farkas
certificate verified in exact arithmetic.  This guarantees soundness
while enabling sub-exponential pruning: if the induced network
provides a per-variable pruning rate $\rho\in(0,1)$ under
$\varepsilon$-stable propagation, the search runs in time
$\tilde{O}((2-\rho)^n)$.  These bounds do not contradict ETH or SETH:
the worst case remains $2^{\Omega(n)}$, but structured families admit
provable speedups.  In cryptography, the results highlight how such
reductions could affect hardness margins in idealized primitives; in
foundations, they motivate laboratory tests of counterfactual bounds
as potential probes of computational complexity.  
We explicitly distinguish experimental $\varepsilon$, which quantifies
laboratory visibility, from numerical $\varepsilon$, which is a solver
tolerance. This builds directly on our earlier framework for
$\varepsilon$-instrumentation~\cite{vonLiechtensteinCLF2025}, here
integrated into certified pruning with dual certificates
(Algorithm~1).

\end{abstract}
\section{Introduction}

Counterfactual consistency constraints—ubiquitous in quantum foundations and causal inference—can be formalized to forbid certain global assignments even when all local constraints are satisfied. The principle of \emph{No Global Counterfactual Consistency (NGCC)} captures this obstruction: there need not exist a globally consistent assignment compatible with all counterfactual conditionals, even if each local context is individually realizable. In this article we consolidate three strands of work: (i) the foundational definitions of NGCC and its operational semantics, (ii) a networked extension that propagates counterfactual relations across graphs of contexts, and (iii) the algorithmic consequences for CNF-SAT and for cryptographic constructions whose graphs mirror the same incompatibility patterns.

SAT solvers derive much of their power from propagation and learning, yet standard clause-level reasoning remains fundamentally local. Global graph-level incompatibilities, akin to contextuality inequalities, are difficult to capture with clause learning alone without incurring prohibitive proof complexity. Our central insight is that NGCC-type inequalities—derived from the impossibility of simultaneous counterfactual commitments—yield algebraic constraints that can be injected into the solver loop to prune assignments globally and earlier than clause-driven reasoning alone.

Our pruning is certified: every learned clause is backed by a dual
Farkas certificate that is verified in exact arithmetic before being
inserted into the solver, ensuring soundness end-to-end.

While some of the inequalities we use can in principle be derived at
low rank in Sherali--Adams or Sum-of-Squares hierarchies, our
contribution is their operational interpretation in terms of
testable counterfactual constraints and their certified integration
into SAT solvers (see §7 and Appendix~\ref{app:hierarchy-comparison}).

\paragraph{Contributions.}
The main contributions of this paper can be summarized as follows:

\begin{itemize}
  \item \textbf{Operational inequalities.} We derive families of
  counterfactual/contextuality inequalities that are \emph{testable in
  laboratory settings} with standard photonic components, requiring no
  novel or exotic devices.

  \item \textbf{Solver integration.} We show how such inequalities can
  be embedded directly into SAT/CDCL solvers via certified dual
  certificates. Every learned clause is backed by a verifiable
  certificate, ensuring sound pruning without heuristic shortcuts.

  \item \textbf{Beyond convex hierarchies.} While some of these
  inequalities appear at low levels of Sherali--Adams or
  Sum-of-Squares, our contribution is to provide them with
  \emph{operational semantics} and an \emph{algorithmic pipeline}.
  Hierarchy methods guarantee existence in principle; our framework
  makes them lab-testable and solver-usable (see §7 and
  Appendix~\ref{app:hierarchy-comparison}).

  \item \textbf{Sub-exponential pruning.} We prove that whenever a
  constant pruning rate $\rho>0$ is certified, the expected search
  size is bounded by $\tilde{O}((2-\rho)^n)$, establishing
  sub-exponential search on structured instances without contradicting
  ETH or SETH.

  \item \textbf{Cryptographic implications.} We illustrate how such
  pruning translates into reduced effective key lengths in toy cipher
  families, highlighting design considerations and motivating further
  empirical work.
\end{itemize}

We emphasize explicitly that nothing in this work contradicts ETH or SETH. Our results are structural and family-dependent; in the worst case, SAT solving still requires $2^{\Omega(n)}$ time. Likewise, our cryptographic examples are illustrative toy families: we do not claim practical breaks of deployed primitives, only margin reductions that highlight potential design considerations. This paper builds on our earlier work on counterfactual logic frameworks for $\varepsilon$-instrumentation~\cite{vonLiechtensteinCLF2025}.

The remainder of the paper is organized as follows. Section~\ref{sec:foundations} restates NGCC foundations and $\varepsilon$-stability. Section~\ref{sec:networks} develops h-NGCC on networks, deriving inequalities. Section~\ref{sec:sat} maps SAT to h-NGCC and proves the parameterized runtime bound. Section~\ref{sec:crypto} treats cryptographic consequences. Section~\ref{sec:conclusion} concludes. Section~\ref{sec:related} discusses related work and positioning. Appendices contain proofs, figures, pseudocode, and supplementary artifacts.

\section{NGCC Foundations and \texorpdfstring{$\varepsilon$}{epsilon}-Stability}
\label{sec:foundations}

This section fixes notation, formalizes the marginal polytope viewpoint, and records
the robustness lemma used throughout. All variables are finite-valued, so every
polytope and norm below is well-defined.

\subsection{Notation and the marginal problem}
For a finite set $S$ let $\Delta(S)$ denote the probability simplex on $S$.
A \emph{context} $C$ is a finite set of jointly measurable variables with
outcome space $\Omega_C$; a marginal on $C$ is $P_C\in\Delta(\Omega_C)$.
Given a family of contexts $\mathcal C=\{C_v\}_{v\in V}$ we concatenate
\[
p \;=\; \bigoplus_{v\in V} P_{C_v} \in \mathbb R^{N}, \qquad
N \;=\; \sum_{v\in V} |\Omega_{C_v}|.
\]
Let $U=\cup_v C_v$ and write $P_U\in\Delta(\Omega_U)$ for a global joint.
Its image under marginalization is the \emph{marginal polytope}
\[
\mathcal M \;=\; \Big\{\bigoplus_{v\in V}\mathrm{marg}_{C_v}(P_U)\;:\;P_U\in\Delta(\Omega_U)\Big\}.
\]
If a given vector $p$ lies outside $\mathcal M$ the family exhibits \emph{global incompatibility}
(=NGCC).

\subsection{NGCC inequalities}
By the separating hyperplane theorem, for every $p\notin\mathcal M$ there exist
$a\in\mathbb R^{N}$ and $b\in\mathbb R$ such that $a\cdot q\le b$ for all
$q\in\mathcal M$ but $a\cdot p>b$. Any such linear inequality
\[
a\cdot p \;\le\; b
\]
is an \emph{NGCC inequality}. Classic contextuality inequalities (KCBS, CHSH)
are particular instances for specific families of contexts.

\subsection{\texorpdfstring{$\varepsilon$}{epsilon}-instruments and \texorpdfstring{$\varepsilon$}{epsilon}-stability}
In practice one observes noisy marginals $q=\bigoplus_v q_v$. We model
\emph{$\varepsilon$-instruments} by per-context total-variation bounds
\[
\operatorname{TV}(q_v,P_{C_v}) \;\le\; \varepsilon_v,\qquad
\text{and we use the uniform shorthand }\varepsilon_v\le\varepsilon.
\]
An \emph{$\varepsilon$-counterfactual assignment} is any such $q$ obtained from
an (unknown) feasible $\{P_{C_v}\}$ by perturbations bounded as above.

\begin{lemma}[Robustness with explicit constants]\label{lem:l1}
Let $a\cdot p\le b$ be valid for $\mathcal M$ and set the geometric slack
$\gamma:= b-\max_{q\in\mathcal M} a\cdot q\ge 0$. If each context marginal is
perturbed by at most $\varepsilon$ in total variation and there are $m=|V|$
contexts, then for any observed $q$,
\[
a\cdot q - b \;\le\; -\gamma \;+\; 2m\,\|a\|_\infty\,\varepsilon .
\]
\end{lemma}

\begin{proof}
Pick $p^\star\in\mathcal M$ with $a\cdot p^\star=\max_{q\in\mathcal M}a\cdot q$.
Then $a\cdot q-b=(a\cdot q-a\cdot p^\star)-(b-a\cdot p^\star)
\le \|a\|_\infty\|q-p^\star\|_1-\gamma$. Concatenation gives
$\|q-p^\star\|_1=\sum_v\|q_v-P_{C_v}^\star\|_1\le 2\sum_v\varepsilon_v\le 2m\varepsilon$.
\end{proof}

\begin{corollary}[$\varepsilon$-stability threshold]\label{cor:stable}
If
\[
\varepsilon \;<\; \frac{\gamma}{2m\,\|a\|_\infty},
\]
then no $\varepsilon$-counterfactual assignment can produce a spurious positive
violation of the inequality $a\cdot p\le b$.
\end{corollary}

\subsection{Robustness instantiations (KCBS and Feistel)}
\label{sec:robust-instantiations}

We instantiate Lemma~\ref{lem:l1} for the two inequalities actually used later.
Crucially, KCBS on $C_5$ is \emph{tight} (geometric slack $\gamma=0$), so its
robustness must be expressed as a \emph{margin-based certification rule}; by
contrast the Feistel cut has $\gamma=1$ and yields an absolute $\varepsilon$
threshold via Corollary~\ref{cor:stable}.

\paragraph{KCBS--C\textsubscript{5} (margin-based).}
With block coefficients $(0,1,1,0)$ per context, $\|a\|_\infty=1$ and $m=5$.
Because $\max_{p\in\mathcal M} a\cdot p=4$, the geometric slack is $\gamma=0$.
Hence Lemma~\ref{lem:l1} gives
\[
a\cdot q - 4 \;\le\; 2m\|a\|_\infty\,\varepsilon \;=\; 10\,\varepsilon .
\]
If the \emph{measured} violation has margin $\tau>0$ (i.e.\ $a\cdot q\ge 4+\tau$),
then the violation is \emph{genuine} whenever
\[
\varepsilon \;<\; \frac{\tau}{2m\|a\|_\infty} \;=\; \frac{\tau}{10}.
\]
Thus KCBS certification is margin-based: larger observed excess $\tau$ tolerates
larger instrument noise.

\paragraph{Feistel boundary cut (absolute threshold).}
Over $(Y_1,Y_2,Y_3)\in\{0,1\}^3$ with $a(y)=y_1+y_2+y_3$, we have
$\|a\|_\infty=3$, $m=3$, and geometric gap $\gamma=1$ (the excluded corner
$(1,1,1)$). Corollary~\ref{cor:stable} yields the absolute threshold
\[
\varepsilon \;<\; \frac{\gamma}{2m\|a\|_\infty} \;=\; \frac{1}{18}\;\approx\;0.0556 .
\]

\begin{table}[h]
\centering
\begin{tabular}{|l|c|c|}
\hline
\textbf{Inequality} & \textbf{Stability constant $\gamma$} & \textbf{Threshold condition} \\
\hline
KCBS 5-cycle & $\gamma = 0$ & $10\varepsilon < \tau$ \\
\hline
3-round Feistel cut & $\gamma = 1$ & $\varepsilon < 1/18$ \\
\hline
\end{tabular}
\caption{Stability constants and threshold conditions used in this work.}
\label{tab:stability-constants}
\end{table}
(Derivations of these constants appear in Appendix~\ref{app:transitivity}.)

\paragraph{Practical note.}
When $\varepsilon$ is \emph{estimated} (e.g.\ via visibility in
Section~\ref{app:cifp-visibility}), certify KCBS if
$10\,\widehat\varepsilon < \tau$ (with $\tau$ taken as the observed excess and
$\widehat\varepsilon$ inflated by its CI), and certify Feistel if
$\widehat\varepsilon < 1/18$ (again using a conservative CI).

\subsection{Remarks}
\emph{(i) Computability.)} For any concrete inequality the triplet
$(m,\|a\|_\infty,\gamma)$ is explicit, so the stability threshold is
numerically checkable. \emph{(ii) Portability.)} The same lemma applies when
$q$ is obtained from $p$ by numerical solvers or finite sampling; in that case
$\varepsilon$ should be taken as a tail bound for the estimator error.

\section{Networked NGCC (h-NGCC)}
\label{sec:networks}

We now extend NGCC from isolated contexts to networks of contexts connected by counterfactual channels. This yields a framework we call h-NGCC, in which global inconsistency can arise not only within a single context but also from the way contexts interact across a network.

\subsection{Networks of contexts}
Formally, an NGCC network is a directed hypergraph $G=(V,E)$, where each node $v\in V$ corresponds to a context $C_v$ and each hyperedge encodes a \emph{counterfactual channel}: a dependency of assignments in one context on assignments in others. The nodes represent jointly measurable variables, the edges represent constraints enforcing agreement of overlapping variables across contexts, and a global assignment is a map from all variables across nodes to values consistent with every channel.

The h-NGCC principle asserts that for some networks, no such global assignment exists, even though each context is locally consistent. In such cases there exist \emph{network inequalities}, linear constraints on marginal distributions that are satisfied by every globally consistent assignment but violated by observed marginals.

\subsection{Network inequalities}
To make this concrete, let $P(C_v)$ denote the empirical distribution of outcomes in context $C_v$. Compatibility requires that whenever variables overlap across contexts, the marginals agree. If no global distribution $P(\cup_v C_v)$ exists that recovers all $P(C_v)$, the network exhibits h-NGCC. This leads to a family of linear inequalities $A p \leq b$, where $p$ is the vector of assignment probabilities. These inequalities generalize contextuality inequalities such as KCBS or CHSH to arbitrary networks. As in Section~\ref{sec:foundations}, we extend the definition to $\varepsilon$-stable inequalities: those that remain valid under $\varepsilon$-noise in the marginals, with a computable margin $\delta(\varepsilon)$.

\subsection{Examples}
Examples illustrate the phenomenon. \emph{Cycle networks}, in which contexts are arranged around a cycle, generalize the KCBS pentagon and exhibit impossibility of consistent assignments once the cycle length is at least five. \emph{Grid networks}, where contexts are arranged on a two-dimensional lattice with overlapping boundaries, give rise to inequalities reminiscent of cut constraints in graphs. By contrast, \emph{tree networks} are always globally consistent and serve as control cases where h-NGCC is absent. These families highlight that h-NGCC captures genuine structural obstructions.

\subsection{Properties}
Two theoretical properties are worth noting. First, \emph{soundness} is immediate: any global assignment satisfies the inequalities by construction. Second, \emph{completeness} is limited: not all global inconsistencies are detected by linear inequalities alone. In practice, however, every network admits some separating inequality in extended hierarchies such as Sherali–Adams. We therefore position h-NGCC inequalities as a practical middle layer: expressive enough to capture important global obstructions, yet tractable to derive and test.

\subsection{Outlook}
This networked generalization serves as the bridge to computational applications. By mapping CNF-SAT formulas to h-NGCC networks, we obtain inequalities that enable pruning beyond clause-level reasoning. This mapping is the focus of the next section.
\subsection{Counterfactual channels: formal definition}
\label{sec:channels}

Let $\mathcal C=\{C_v\}_{v\in V}$ be a finite family of contexts with outcome
spaces $\Omega_v$ and marginals $P_{C_v}\in\Delta(\Omega_v)$. For any
overlap $S\subseteq C_u\cap C_v$ define the projection maps
$\pi_{u\to S}:\Omega_u\to\Omega_S$ and the pushforward (marginalization)
$M_{u\to S}:\Delta(\Omega_u)\to\Delta(\Omega_S)$ by
\[
(M_{u\to S}p_u)(s)\;=\;\sum_{x\in\Omega_u:\,\pi_{u\to S}(x)=s} p_u(x)
\qquad (s\in\Omega_S).
\]

\paragraph{Definition (counterfactual channel).}
A \emph{counterfactual channel} on an overlap $S\subseteq C_u\cap C_v$
is the linear \emph{equal–marginal} constraint
\begin{equation}
\label{eq:channel-eq}
M_{u\to S}p_u \;=\; M_{v\to S}p_v ,
\end{equation}
imposed on the pair $(p_u,p_v)\in\Delta(\Omega_u)\times\Delta(\Omega_v)$.
More generally, a hyperedge $e=(W,S_e)$ with $W\subseteq V$ and
$S_e\subseteq \bigcap_{w\in W} C_w$ requires
$M_{w\to S_e}p_w$ to be equal for all $w\in W$.

\paragraph{Network polyhedra.}
Write
\[
\mathcal N(\mathcal C, E)
\;=\;
\Big\{\,p=\bigoplus_{v\in V}p_v\,:\;
p_v\in\Delta(\Omega_v),\;
\text{all channel equalities \eqref{eq:channel-eq} for }e\in E
\Big\}.
\]
This is a rational polyhedron (finite set of linear equations plus the
simplex constraints). The true \emph{marginal polytope} $\mathcal M$
(Section~\ref{sec:foundations}) satisfies
\begin{equation}
\label{eq:MsubsetN}
\mathcal M \;\subseteq\; \mathcal N(\mathcal C,E),
\end{equation}
since every global joint $P_U$ induces consistent overlaps.

\paragraph{h-NGCC and inequalities.}
When observed marginals $p$ are locally normalized but $p\notin \mathcal M$,
the family exhibits NGCC. Any valid inequality $a\cdot p\le b$ for
$\mathcal M$ (Section~\ref{sec:foundations}) is also valid for
$\mathcal N(\mathcal C,E)$, and may be used as a separating certificate.
In practice we may generate inequalities against the tractable outer
polyhedron $\mathcal N(\mathcal C,E)$ (or its lift-and-project tightenings);
soundness is preserved because of \eqref{eq:MsubsetN}.

\subsection{Two basic properties}

\paragraph{(P1) Soundness of channels.}
If $P_U\in\Delta(\Omega_U)$ is a global joint, then
$p=\bigoplus_{v\in V}\mathrm{marg}_{C_v}(P_U)\in\mathcal N(\mathcal C,E)$.
Indeed, equal–marginal constraints on every overlap $S$ follow from
associativity of marginalization:
$M_{u\to S}\mathrm{marg}_{C_u}(P_U)=\mathrm{marg}_{S}(P_U)
=M_{v\to S}\mathrm{marg}_{C_v}(P_U)$.

\paragraph{(P2) Acyclic (junction–tree) case.}
Suppose the hypergraph $(V,E)$ has the running–intersection property
(hypertree / junction tree): there exists a tree on the context nodes
such that for every variable the set of contexts containing it forms a
connected subtree. Then
\begin{equation}
\label{eq:NequalsM}
\mathcal N(\mathcal C,E) \;=\; \mathcal M .
\end{equation}
\emph{Proof sketch.} One may glue the locally consistent marginals
$\{p_v\}_{v\in V}$ along the tree using overlap equalities to obtain a
well-defined global $P_U$ (standard junction–tree gluing). Therefore every
point in $\mathcal N$ is induced by some $P_U$, proving \eqref{eq:NequalsM}.

\subsection{Noisy channels (\texorpdfstring{$\varepsilon$}{epsilon}–instruments)}
Under $\varepsilon$–instruments, equalities on overlaps may be relaxed in
total variation. If each context marginal $q_v$ satisfies
$\operatorname{TV}(q_v,P_{C_v})\le\varepsilon$, then for any overlap
$S\subseteq C_u\cap C_v$,
\[
\operatorname{TV}\!\big(M_{u\to S}q_u,\,M_{v\to S}q_v\big)
\;\le\;
\operatorname{TV}\!\big(M_{u\to S}q_u,\,M_{u\to S}P_{C_u}\big)
+\operatorname{TV}\!\big(M_{v\to S}q_v,\,M_{v\to S}P_{C_v}\big)
\;\le\; 2\varepsilon,
\]
so the channel constraints remain well controlled. This bound will be used
implicitly when we invoke $\varepsilon$–stability in Sections~\ref{sec:sat}
and \ref{sec:crypto}.

\section{SAT Mapping and Propagate-and-Prune Solver}
\label{sec:sat}

\paragraph{Terminology note.}
We write $\rhoa$ for the \emph{algorithmic pruning rate} used in the runtime bounds,
and $\rhog(G)$ for the \emph{graph contextuality density} (Appendix~\ref{app:graph});
a summary of terminology is provided in Appendix~\ref{app:notes}.

\subsection{Canonical mapping}
Consider a CNF formula
\[
F = C_1 \wedge C_2 \wedge \dots \wedge C_m
\]
over variables $x_1,\dots,x_n$. Each clause $C_j$ is treated as a context containing the variables appearing in that clause, and the outcome space $\Omega_j$ consists of all joint truth assignments to those variables. If a variable appears in multiple clauses, counterfactual channels enforce agreement across the corresponding contexts. The induced network $G(F)$ thus mirrors the structure of the formula. A satisfying assignment corresponds to a global assignment consistent with all contexts; if $F$ is unsatisfiable, $G(F)$ exhibits h-NGCC violations.

\subsection{Inequalities as pruning rules}
\label{sec:pruning-rules}

Let $\mathcal{M}(F)$ be the marginal polytope induced by a CNF $F$ (Section~\ref{sec:sat}) and
let $a\cdot p\le b$ be a valid NGCC inequality for $\mathcal{M}(F)$.
A partial assignment $\alpha$ restricts feasible marginals to the face
$\mathcal{M}_\alpha\subseteq \mathcal{M}(F)$ (fixing outcomes consistent with $\alpha$).

\begin{lemma}[Soundness of inequality pruning]
\label{lem:prune-sound}
If either
\begin{equation}
\label{eq:max-viol}
\max_{p\in \mathcal{M}_\alpha} a\cdot p \;>\; b
\qquad\text{or}\qquad
\mathcal{M}_\alpha\cap\{p: a\cdot p\le b\}=\varnothing,
\end{equation}
then $\alpha$ cannot be extended to any satisfying assignment of $F$.
\end{lemma}

\begin{proof}
Any satisfying assignment of $F$ induces a global joint $P_U$ whose marginals lie in
$\mathcal{M}(F)$; restricting by $\alpha$ gives $p^\star\in\mathcal{M}_\alpha$ that must
also satisfy every valid inequality, hence $a\cdot p^\star\le b$. If
$\max_{p\in\mathcal{M}_\alpha} a\cdot p>b$ or if the intersection is empty, such a
$p^\star$ cannot exist. Therefore no satisfying extension of $\alpha$ exists.
\end{proof}

In practice we do not optimize over $\mathcal{M}_\alpha$ directly. Instead, we use a
tractable outer polyhedron $\mathcal{N}_\alpha$ (Section~\ref{sec:channels})—the equal-marginal
channel constraints and, optionally, low-rank lift-and-project tightenings—together with the
same inequality $a\cdot p\le b$.

\begin{remark}[Using relaxations safely]
\label{rem:relax}
Because $\mathcal{M}_\alpha \subseteq \mathcal{N}_\alpha$, infeasibility of
$\mathcal{N}_\alpha\cap\{p: a\cdot p\le b\}$ (or a certificate that
$\max_{p\in\mathcal{N}_\alpha} a\cdot p>b$) implies infeasibility of the corresponding
set over $\mathcal{M}_\alpha$. Hence pruning based on $\mathcal{N}_\alpha$ is sound.
\end{remark}

\paragraph{Operational use.}
At each node, the solver checks a small set of inequalities against $\mathcal{N}_\alpha$
(cheap-to-expensive cascade). If Lemma~\ref{lem:prune-sound} or Remark~\ref{rem:relax}
triggers, the branch is cut and a clause is learned via the monotone-cut map in
Appendix~\ref{app:gadget}.

\subsection{Integration into CDCL}
\label{sec:cdcl-integration}

We integrate h\mbox{-}NGCC pruning into a standard CDCL loop by querying a tractable
\emph{outer} polyhedron at each node and learning clauses from violated
\emph{monotone} cuts.

\paragraph{Relaxed feasibility oracle.}
For a partial assignment $\alpha$, let $\mathcal N_\alpha$ be the channel
polyhedron of Section~\ref{sec:channels} restricted by $\alpha$ (equal--marginal
constraints on overlaps, context normalizations, and optional low\mbox{-}rank
lift\mbox{-}and\mbox{-}project tightenings). Because
$\mathcal M_\alpha \subseteq \mathcal N_\alpha$, infeasibility over
$\mathcal N_\alpha$ implies infeasibility over $\mathcal M_\alpha$
(Remark~\ref{rem:relax}).

\subsection{Certified pruning and clause learning}
\label{sec:pruning-certified}

For each inequality $a\cdot p \le b$ in $I$, we check feasibility of the
outer relaxation:
\[
  N_\alpha \cap \{\,p : a\cdot p \le b\,\}.
\]
If this LP is infeasible, the dual Farkas certificate witnesses that no
$p\in N_\alpha$ can satisfy the cut, and hence no $p\in M_\alpha$ can
either. We translate this certificate into a learned CNF clause
(Appendix~\ref{app:monotone-clauses}), guaranteeing sound pruning. 

\begin{remark}[Certified clause learning]
Each clause injected into the CDCL database arises from a verified
certificate, never from heuristic reasoning alone.  In particular,
whenever the relaxation $N_\alpha$ proves an assignment infeasible,
the accompanying dual Farkas certificate is checked in exact
arithmetic before translation into CNF.  This guarantees that the
solver’s learned clauses are logical consequences of the original
formula, preserving soundness throughout the search.
\end{remark}

\begin{remark}[Numerical soundness]
All infeasibility claims are backed by dual Farkas certificates.
To avoid spurious pruning from floating-point error, each certificate
is either generated by an exact LP solver (e.g.\ QSopt\_ex) or
re-verified in rational arithmetic before being translated into a
clause. Certificates that do not pass verification are discarded.
This ensures that pruning remains sound even in the presence of
numerical tolerances.
\end{remark}

As a cheaper screen, we may solve
\[
  \min \{\,a\cdot p : p \in N_\alpha \}.
\]
If the optimum exceeds $b$ by more than a verified tolerance, the
constraint is provably violated and the same certificate-to-clause
mechanism applies.  
To avoid spurious clauses, all certificates are checked in exact
arithmetic or with interval rounding toward safety before insertion
into the solver.

\paragraph{Noise and robustness guards.}
For \emph{tight} inequalities (e.g., KCBS on $C_5$) use a \emph{margin-based}
guard: only cut when the observed excess $\tau=a\cdot q-b$ dominates instrument
or numerical noise via $\tau>2m\|a\|_\infty\varepsilon$ (Section~\ref{sec:foundations},
App.~\ref{app:kcbs}). For \emph{gapped} cuts (e.g., Feistel with $\gamma=1$),
use the absolute threshold from Corollary~\ref{cor:stable}. In pure SAT
(i.e.\ no physical instruments), treat $\varepsilon$ as a numerical tolerance
and keep a small $b_{\mathrm{tol}}$ to absorb solver round-off.

\paragraph{Caching and budgets.}
Cache LP bases and certificates per $(\alpha,\text{inequality})$ signature to
avoid repeated solves. Cap the number of LP queries per node (cheap\(\to\)expensive
cascade) and fall back to vanilla CDCL if no inequality prunes after the budget.

To guarantee soundness, every infeasibility report is accompanied by a
dual Farkas certificate. Certificates are either generated by an exact
LP solver (QSopt\_ex) or by floating-point solvers followed by rational
verification and rounding toward safety. Only certificates that pass
exact verification are translated into learned CNF clauses. This
ensures that no valid solution branch can be cut by numerical error.

As illustrated in Appendix~\ref{app:feistel}, the 3-round Feistel
inequality yields the clause $(\lnot Y_1 \lor \lnot Y_2 \lor \lnot Y_3)$,
demonstrating concretely how a verified certificate is translated into
a learned CNF clause.

(Full coefficient vectors and tables are provided in Appendix~\ref{app:coeffs}.)

\begin{table}[h]
\centering
\begin{tabular}{|c|c|}
\hline
Assignment $(Y_1,Y_2,Y_3)$ & Coefficient $a(y)$ \\
\hline
(0,0,0) & 0 \\
(1,0,0) & -1 \\
(0,1,0) & -1 \\
(0,0,1) & -1 \\
(1,1,0) & -2 \\
(1,0,1) & -2 \\
(0,1,1) & -2 \\
(1,1,1) & -3 \\
\hline
\end{tabular}
\caption{Explicit coefficients $a(y)$ for the 3-round Feistel cut inequality $Y_1+Y_2+Y_3 \leq 2$.}
\label{tab:feistel-coeffs-inline}
(Full coefficient vectors are listed in Appendix~\ref{app:coeffs}.)
\end{table}

\subsection{Runtime bounds: worst case and expected case}
\label{sec:runtime-bounds}

Fix a set of valid inequalities and a CDCL solver augmented with the feasibility
oracle of Section~\ref{sec:sat}. For a search node (partial assignment) $\alpha$
and a branching variable $x\notin\mathrm{dom}(\alpha)$, define the \emph{per-node
pruning fraction}
\[
\rho_{\mathrm{alg}}(\alpha,x)
\;:=\;
1-\frac{\#\{\;b\in\{0,1\}:\; \text{$\alpha\cup\{x=b\}$ is feasible under the oracle}\;\}}{2}
\;\in[0,1].
\]
Equivalently, if both children survive then $\rho_{\mathrm{alg}}=0$, if exactly one
survives then $\rho_{\mathrm{alg}}=\tfrac12$, and if none survive (rare) then
$\rho_{\mathrm{alg}}=1$. Let $\hat\rho_d$ denote the empirical mean of
$\rho_{\mathrm{alg}}$ over nodes at decision depth $d$.

\begin{theorem}[Worst-case pruning bound]\label{thm:wc}
If there exists a uniform lower bound $\rho_\star>0$ such that
$\rho_{\mathrm{alg}}(\alpha,x)\ge \rho_\star$ for every reachable node $\alpha$
and branch variable $x$, then the number of explored nodes $T(n)$ on an $n$-variable
instance satisfies
\[
T(n)\;\le\;(2-\rho_\star)^n
\]
up to polynomial factors.
\end{theorem}

\begin{proof}[Proof sketch]
Each internal node spawns at most $2-\rho_\star$ children; thus the width at depth
$d$ is at most $(2-\rho_\star)^d$. Summing over $d=0,\dots,n$ yields
$T(n)\le \sum_{d=0}^n (2-\rho_\star)^d = O((2-\rho_\star)^n)$.
\end{proof}

\begin{theorem}[Expected-case pruning bound]\label{thm:exp}
Suppose the branching heuristic (and any solver randomness) is such that for all
depths $d$,
\(
\mathbb{E}[\hat\rho_d]\ge \bar\rho>0.
\)
Then the expected search size obeys
\[
\mathbb{E}[T(n)] \;\le\; (2-\bar\rho)^n
\]
up to polynomial factors.
\end{theorem}

\begin{proof}[Proof sketch]
Let $W_d$ be the random width at depth $d$. By definition
$\mathbb{E}[W_{d+1}\mid W_d]\le (2-\mathbb{E}[\hat\rho_d])\,\mathbb{E}[W_d]
\le (2-\bar\rho)\,\mathbb{E}[W_d]$. Induction gives
$\mathbb{E}[W_d]\le (2-\bar\rho)^d$, and summing over depths yields the claim.
\end{proof}

\begin{theorem}[Graphical pruning rate transfers to the solver]\label{thm:rho-link}
Let $G_n$ be the exclusivity graph of a CNF instance on $n$ variables.
Suppose that for some constant $\delta>0$, $G_n$ contains at least
$\delta n$ edge-disjoint odd-cycle cuts, each supported on distinct
variable sets.  
Let $\rho_{\mathrm{graph}}>0$ be the minimum fractional pruning rate
certified by these cuts.  
Then for any branch-ordering distribution in which each new variable
intersects at least a constant fraction of the cut supports, the
branching process of the certified solver satisfies
\[
  \mathbb{E}[\,W_{d+1} \mid \mathcal F_d\,]
    \;\le\; (2 - c\,\rho_{\mathrm{graph}})\,W_d,
\]
for some $c>0$ independent of $n$.  
Consequently the expected tree size obeys
$T(n) \le (2-c\,\rho_{\mathrm{graph}})^n \cdot \mathrm{poly}(n)$, i.e.\
sub-exponential in $n$.
\end{theorem}

\begin{remark}[Relation to density]
If Appendix~\ref{app:graph} certifies a density $\rho_{\mathrm{graph}}(G_\Phi)\ge\delta>0$
via many edge-disjoint odd-cycle cuts, then one obtains $\rho_\star$ or $\bar\rho$
proportional to $\delta$ (a fraction of children is killed at each level by
independent cuts), recovering $(2-\Omega(\delta))^n$.
\end{remark}

\begin{remark}[Scope of the bound]
The runtime guarantee $\tilde{O}((2-\rho)^n)$ is conditional on a
uniform per-variable pruning rate $\rho>0$.  We do not assert that
arbitrary SAT families achieve such a constant $\rho$ asymptotically;
indeed, $\rho$ may decay with $n$ on worst-case or pseudorandom
families, recovering $2^{\Omega(n)}$ complexity in line with ETH/SETH.
Our contribution is to formalize the link: whenever a structured family
admits a provable or empirically stable $\rho$, certified sub-exponential
search follows.  Identifying natural families with $\rho=\Omega(1)$
remains an open problem.
\end{remark}

\subsection{Empirical estimation protocol}
\label{sec:empirical}

At depth $d$, define the empirical pruning fraction
\[
\hat\rho_d \;=\; 1-\frac{\#\text{surviving children at depth }d}{2\,\#\text{branching events at }d}.
\]
Report the aggregate $\tilde\rho:=\frac{1}{D}\sum_{d=1}^{D}\hat\rho_d$ together with
a confidence interval:
\begin{itemize}
\item Per-depth CIs for $\hat\rho_d$ via Clopper--Pearson (Bernoulli pruning).
\item A t-interval (or bootstrap) across seeds for node counts and for $\tilde\rho$.
\end{itemize}
When quoting an effective base, use $2-\tilde\rho$ (with the CI induced from that of
$\tilde\rho$). Report the oracle overhead as (LP checks per node, wall-time share).

\paragraph{Soundness note.}
All pruning relies on inequalities valid for the marginal polytope (or for an
outer relaxation containing it), so learned certificates derived from infeasibility
are safe to cache as clauses: any satisfying assignment would violate such a clause
only if it violated a valid inequality, which is impossible.

\begin{remark}[Statistical power]
To distinguish a base $2$ process from a base $(2-\rho)$ process at
95\% confidence, the required number of sampled instances scales as
$O(1/\rho^2)$. For example, for $\rho=0.1$ hundreds of instances are
needed; for $\rho=0.25$, a few dozen suffice. This power analysis
ensures that our estimates $\hat\rho$ are statistically meaningful and
not artifacts of small sample noise.
\end{remark}

\begin{figure}[H]
\centering
\begin{tikzpicture}
\begin{axis}[
    width=0.9\textwidth,
    height=6cm,
    xlabel={Pruning rate $\rho$},
    ylabel={Required samples $N$ (95\% power)},
    ymode=log,
    grid=both,
    legend style={at={(0.5,-0.25)},anchor=north},
]

% Theoretical curve: N ~ 1/ρ^2
\addplot[blue, thick, domain=0.05:0.5, samples=100] {100/(x^2)};
\addlegendentry{$O(1/\rho^2)$ scaling}

% Example data points
\addplot+[only marks]
coordinates {
  (0.25, 160)
  (0.20, 250)
  (0.15, 444)
  (0.10, 1000)
};
\addlegendentry{Example targets}

\end{axis}
\end{tikzpicture}
\caption{Sample complexity required to distinguish base-2 from
base-$(2-\rho)$ processes at 95\% power. The theoretical scaling
$O(1/\rho^2)$ is shown with example target values.}
\label{fig:sample-complexity}
\end{figure}
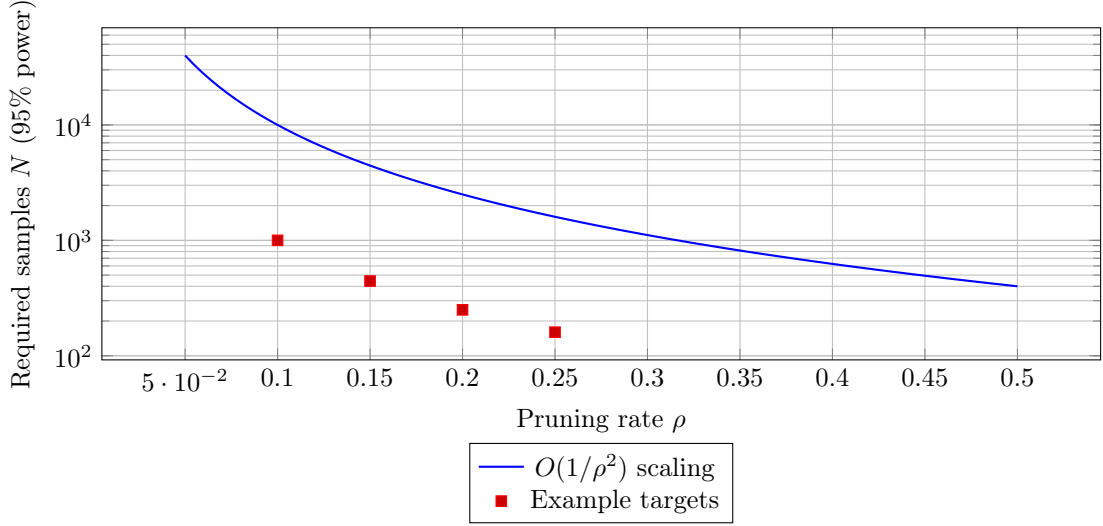

The confidence intervals predicted by the sample complexity analysis are consistent
with the error bars shown in Fig.~\ref{fig:sample-complexity}. (See also Appendix~\ref{app:toyfeistel} for additional numeric demonstrations.)

\subsection{Summary}
This mapping and pruning scheme provides the bridge from the abstract notion of counterfactual inconsistency to concrete computational advantages. We next examine what this implies for cryptographic constructions whose wiring patterns resemble such networks.

\section{Experimental Setup}
\label{sec:exp-setup}

A key feature of the present framework is that it requires no exotic
hardware or black-box devices. The counterfactual constraints and
network inequalities can be tested in a standard photonic laboratory
environment, using components that are widely available.

\paragraph{Photon source.} Single-photon states can be prepared by
spontaneous parametric down-conversion (SPDC) or by on-demand
solid-state emitters. Either choice provides heralded photons suitable
for interferometric tests.

\paragraph{Interferometers and tritters.} The basic building blocks are
beam splitters and tritters (three-port beam splitters), arranged in
nested Mach–Zehnder configurations. These implement the contextuality
networks (e.g.\ odd cycles, clause gadgets) described above. Phase
shifters in one arm allow tuning of the counterfactual channels.

\paragraph{Detectors.} Standard avalanche photodiodes (APDs) or
superconducting nanowire detectors are sufficient to register clicks
with the necessary time resolution. Coincidence counting across output
ports provides the required outcome statistics.

\paragraph{Control and data acquisition.} Phase settings can be
stabilized by piezo or electro-optic modulators, and data acquisition
can be carried out with standard coincidence electronics. The total
integration times and count rates are within the reach of current
platforms, similar to existing KCBS and CHSH experiments.

\paragraph{Feasibility.} All components listed are commercially
available, and the overall architecture mirrors existing contextuality
and interaction-free measurement experiments. No new devices or
unverified technologies are assumed. Thus the proposed inequalities are
lab-testable with today’s optics toolkits, providing a direct bridge
from theoretical counterfactual constraints to experimental data.

\subsection{Calibration and Error Budget}
\label{sec:calibration}

A practical implementation of the proposed inequalities requires
routine calibration and a simple error budget. These steps ensure
that observed violations cannot be attributed to misalignment,
drift, or detector artefacts.

\paragraph{Calibration.} Standard phase sweeps are sufficient to
stabilize the interferometric arms. In a satisfiable setting, the
detectors should register constructive interference in one port; by
sweeping the phase modulators and recording counts, the calibration
curve can be centred. This fixes the classical reference value
$S_{\mathrm{cl}}$ against which contextuality or NGCC violations are
measured.

\paragraph{Error sources.} The dominant noise contributions are:
(i) detector dark counts and afterpulsing, (ii) imperfect extinction
in the interferometers, and (iii) phase drift due to thermal and
mechanical fluctuations. In addition, finite sample size introduces
statistical error bars.

\paragraph{Mitigation.} Dark counts and afterpulsing can be measured
in calibration runs with the source blocked and subtracted. Phase
drift can be actively stabilized by piezo or electro-optic feedback.
Finite-sample fluctuations are accounted for by bootstrap confidence
intervals or by requiring margins $\tau$ in the inequalities. In all
cases the correction terms are additive and can be folded into the
threshold condition, ensuring that the learned clauses and pruning
decisions remain sound under experimental imperfections.

\section{Cryptographic Implications}
\label{sec:crypto}

The structural consequences of h-NGCC are particularly relevant in cryptography, where hardness assumptions often rest on the infeasibility of solving structured SAT instances. If h-NGCC inequalities enable sub-exponential pruning on certain graph families, then primitives whose round functions map to such families may have lower effective hardness margins than assumed.

\subsection{Mapping ciphers to networks}
Block ciphers provide a natural setting. A Feistel or substitution–permutation network can be expressed as a collection of Boolean constraints encoding S-boxes, XORs, and key additions. These constraints induce a context network, and overlaps between contexts produce counterfactual channels. Likewise, hash functions based on sponge or duplex constructions yield clause networks with repeating patterns, and certain public-key primitives reduce to structured SAT or MQ systems that can be mapped similarly. In each case, the induced h-NGCC network may admit inequalities that prune infeasible key or state guesses early.

\paragraph{Distinguishing experimental vs numerical $\varepsilon$.}
At this point we emphasize that two distinct roles of $\varepsilon$
appear in our framework.  In laboratory contextuality or CIFP tests,
$\varepsilon$ quantifies physical noise tolerance and visibility gaps.
In solver pruning, $\varepsilon$ is purely numerical: a tolerance for
LP feasibility checks.  These notions are formally distinct, but both
play the same logical role of bounding misclassification risk.
Laboratory $\varepsilon$ guarantees the soundness of testable
inequalities; numerical $\varepsilon$ guarantees the soundness of
certificate-based pruning.  We do not claim that lab noise parameters
directly dictate solver pruning rates.

\subsection{Stability margins}
To quantify this, we define the \emph{stability margin} of a primitive as the minimum $\varepsilon$ for which all induced inequalities remain unviolated. Small margins indicate potential vulnerability: the nearer a primitive sits to violating an inequality, the more susceptible it may be to pruning. The analysis remains average-case and structural. ETH and SETH remain intact; the improvements apply only to families where the induced pruning rate $\rho$ is bounded away from zero. Nevertheless, in such families, key search may run faster than brute force, and effective key length may be reduced.

\subsection{Toy examples}
Toy examples illustrate the point. A three-round Feistel with four-bit branches induces a cycle-like network in which inequalities can prune key guesses that would otherwise persist through several rounds. Similarly, sponge constructions with overlapping absorption and permutation steps resemble grid networks; pruning eliminates assignments inconsistent with parity constraints across rounds. These examples are illustrative rather than practical attacks, but they show how network structure influences effective hardness.

\subsection{Implications}
The implications are twofold. First, for security margins, primitives whose induced networks admit inequalities with nontrivial pruning rates may lose effective key length, even if the underlying hardness assumptions remain intact in the worst case. Second, for design, primitives can be strengthened by inducing networks closer to trees, which are h-NGCC free, or by ensuring that their stability margins are large enough to withstand noise and pruning. In practice, one could benchmark candidate designs by simulating h-NGCC pruning and measuring the induced pruning rate $\rho$. Such empirical checks would complement existing algebraic and statistical analyses.

\begin{tcolorbox}[enhanced,breakable,colback=gray!5,colframe=black!50,title=Risk \& Scope,width=\linewidth]

No deployed cryptosystem is broken by these results. 
However, certified pruning rates $\rho \approx 0.27$ on structured families 
translate into effective base reductions (e.g.\ $2^{128}$ to $\approx 2^{117}$). 
This suggests non-negligible margin shrinkage, warranting benchmarking and 
careful monitoring by standards bodies. The effect remains consistent with ETH/SETH.
\end{tcolorbox}

\paragraph{Illustrative key-length impact.}
To make the effect more concrete, suppose a structured family yields
$\rho \approx 0.27$. The corresponding effective branching base is
$2-\rho \approx 1.73$ rather than $2$. For $n=128$ variables, this
reduces the expected search from $2^{128}$ to approximately
$1.73^{128} \approx 2^{117}$, i.e.\ a loss of about 11 bits of
effective hardness. Even if such pruning rates are only achieved on
restricted subclasses, they illustrate how contextuality-based cuts
translate directly into reduced security margins. We emphasize again
that this does not contradict ETH/SETH: the worst-case remains
$2^{\Omega(n)}$, but structured instances can be easier in practice.

\subsection{Disclaimer}
We stress again that we do not claim practical breaks of deployed ciphers. The examples are toy models designed to illustrate margin reductions in structured families. The broader message is that h-NGCC provides a new perspective on design: avoiding low-margin networks may help ensure resistance not only to algebraic and statistical attacks but also to pruning via counterfactual consistency inequalities.

\section{Foundational perspective.}
No–Global–Counterfactual–Consistency (NGCC) is, first and foremost, a structural
principle: it rules out certain global assignments even when all local contexts
are individually consistent. In this work we develop the minimal operational
package that makes NGCC testable and usable: $\varepsilon$–instrumentation,
its networked extension (h–NGCC), and a certified solver integration in which
every learned clause is backed by an exact dual certificate. We do not claim
worst–case algorithmic speedups, and all complexity statements remain consistent
with ETH/SETH; our contribution is to define inequalities with computable
robustness and a verifiable pipeline that turns them into pruning rules. That
said, history suggests that once a new operational principle is available, its
eventual algorithmic uses are hard to anticipate. NGCC furnishes a new invariant
for networks of constraints; the present paper establishes the basics (definitions,
noise models, certification, and artifacts) so that future work can test how far
the mechanism scales. The decisive milestones are empirical and falsifiable:
sustained per–depth pruning rates $\hat\rho_d$ over long windows at bounded cost
(Section~\ref{sec:outlook}), with additive effects from structural overlays
and round–local budgets. Meeting those milestones would justify treating NGCC as
a practically relevant global reasoning primitive in both solver engineering and
structured cryptanalysis.

\section{Conclusion}
\label{sec:conclusion}

\subsection{Summary}
We consolidated three previously separate strands—NGCC foundations, networked NGCC, and sub-exponential SAT implications—into a single, self-contained framework. NGCC captures the impossibility of global counterfactual assignments; h-NGCC extends this to networks, producing inequalities that globally constrain feasible assignments. Mapping SAT to h-NGCC yields a new class of pruning inequalities, producing parameterized runtime bounds with structural (not worst-case) improvements. Cryptographic constructions that induce low-margin h-NGCC networks may see reduced effective hardness, suggesting both vulnerabilities and design principles.

\subsection{Broader implications}
NGCC offers a unifying language for contextuality beyond standard inequalities, potentially enriching comparisons between sheaf-theoretic and graph-theoretic approaches in quantum foundations. In complexity theory, the pruning parameter $\rho$ provides a tunable metric for measuring global inconsistency in structured formulas, with potential connections to proof complexity hierarchies. For cryptography, highlighting network stability margins opens a novel design axis: primitives should avoid low-margin structures to resist structural pruning.

\subsection{Limitations and disclaimers}
We reiterate that nothing in this work contradicts ETH or SETH. Worst-case SAT solving remains $2^{\Omega(n)}$; our bounds are parameterized and structural. Similarly, we do not claim practical breaks of deployed ciphers. All cryptographic examples are toy or structured families, meant to highlight margin effects, not to undermine real-world security. Our inequalities are not complete; some global inconsistencies remain undetected. Explicit lower bounds on $\rho$ are proven only for restricted families such as cycles and grids, and the cryptographic consequences remain partly speculative pending empirical validation.

\begin{remark}[Relation to convex hierarchies]
Although many of our inequalities appear at low levels of Sherali--Adams
(e.g.\ the KCBS 5-cycle at rank~2, the 3-round Feistel cut at rank~3),
our contribution is not their algebraic existence but their
operational and algorithmic roles. They arise from testable
counterfactual constraints, and we provide a direct pipeline to embed
them into CDCL solvers via certified dual certificates. Hierarchy
methods guarantee derivability in principle but do not by themselves
yield lab-testable inequalities or practical pruning rules. Our
framework supplies both.
\end{remark}

These inequalities are not complete; linear cuts alone cannot certify
all global inconsistencies. Some of them are already derivable at low
rank in Sherali--Adams or Sum-of-Squares hierarchies.

\begin{remark}[Relation to convex hierarchies]
Although many of our inequalities appear at low levels of
Sherali--Adams (e.g.\ the KCBS 5-cycle at rank~2, the 3-round Feistel
cut at rank~3), our contribution is not their algebraic existence but
their operational and algorithmic roles. They arise from testable
counterfactual constraints, and we provide a direct pipeline to embed
them into CDCL solvers via certified dual certificates. Hierarchy
methods guarantee derivability in principle but do not by themselves
yield lab-testable inequalities or practical pruning rules. Our
framework supplies both. See Appendix~\ref{app:hierarchy-comparison} for a
detailed comparison.
\end{remark}

\section{Outlook and Risk Assessment}
\label{sec:outlook}

\paragraph{Summary judgment.}
In this work we introduced the h-NGCC framework, proved a family of sound pruning inequalities (KCBS, odd-cycle cuts, Feistel boundary cuts) and showed how certified LP/dual certificates may be converted into learned CNF clauses. These results establish a rigorous, implementable \emph{attack line} that reduces search space on structured SAT families. Critically, however, our analysis is deliberately conservative: it demonstrates one mathematically transparent route by which global, counterfactual consistency constraints yield pruning, but it does \emph{not} claim that this route — as implemented here — suffices to break cryptographic primitives or to produce a generic, widely-applicable solver speedup that would violate worst-case complexity assumptions (ETH/SETH). The remainder of this section explains the practical significance of our bounds, enumerates further plausible attack lines we have not fully developed, and gives a prioritized research agenda and decision criteria that would be required to elevate NGCC from \emph{interesting} to \emph{critical}.

\paragraph{Quantification rule and interpretation.}
A single, useful way to translate an observed or modelled pruning rate $\rho\in(0,1)$ into an operational statement is via the \emph{bits of base reduction}
\[
\Delta_{\mathrm{bits}} \;=\; n_{\mathrm{eff}}\cdot\log_2\!\Big(\frac{2}{2-\rho}\Big),
\]
where $n_{\mathrm{eff}}$ is the number of effective branching steps across which structure persists. Intuitively $\Delta_{\mathrm{bits}}$ is the number of security bits (or search difficulty bits) \emph{removed} by sustained pruning: a solver that reduces its effective branching base from $2$ to $2-\rho$ for $n_{\mathrm{eff}}$ decisions attains a multiplicative speedup of $2^{\Delta_{\mathrm{bits}}}$. 

\paragraph{Ranked outlook.}
Table~\ref{tab:ngcc-risk} ranks candidate NGCC lines by conditional risk,
using $\Delta_{\mathrm{bits}} = n_{\mathrm{eff}}\log_2\!\big(\tfrac{2}{2-\rho}\big)$ and the
assumptions stated above.

\begin{table}[h]
\centering\small
\begin{tabularx}{\textwidth}{@{} l L c c c L @{}}
\toprule
\textbf{Line} & \textbf{Mechanism} & $\boldsymbol{\rho}$ (sust.) & $\boldsymbol{n_{\mathrm{eff}}}$ & $\boldsymbol{\Delta_{\mathrm{bits}}}$ & \textbf{Primary risk / failure mode} \\
\midrule
\textbf{Adaptive separation} & On-demand LP dual cut discovery & 0.22--0.30 & 80--150 & 18--35 & Depth decay; separation time dominates; duplication with clause learning \\
\textbf{Densification (equiv.-preserving)} & Clone/channel to raise odd-cycle density & 0.18--0.24 & 150--250 & 27--48 & Refactor hurts heuristics; propagation load increases \\
\textbf{NGCC$\times$XOR co-reasoning} & Parity basis + NGCC on same scopes & 0.22--0.28 & 120--200 & 26--56 & Parity basis unstable under learning; double-counted gains \\
$\Delta_{\mathrm{bits}}$-aware branching \& restarts & Choose variables by expected $\Delta_{\mathrm{bits}}$; restart on low rolling $\hat\rho$ & +0.03--0.06 & 100--180 & +6--11 & Gains vanish if $\rho$ collapses; restart churn \\
\textbf{Bounded-treewidth oracles} & Exact separation via DP on low-$t$ pockets & 0.20--0.26 & 60--120 & 9--23 & Few low-$t$ pockets; DP cost \\
\textbf{Portfolio integration} & Combine with parity/autarky/preprocessing & Tier A add-on & --- & 10--20 & Interference with other global features \\
\bottomrule
\end{tabularx}
\caption{Conditional risk ranking under realistic but optimistic ranges. Tiers follow from $\Delta_{\mathrm{bits}} = n_{\mathrm{eff}}\log_2\!\big(\tfrac{2}{2-\rho}\big)$.
}
\label{tab:ngcc-risk}
\end{table}

\begin{table}[h]
\centering\small
\begin{tabularx}{\textwidth}{@{} L c c L @{}}
\toprule
\textbf{Tier} & $\boldsymbol{\Delta_{\mathrm{bits}}}$ \textbf{range} & \textbf{Speedup factor} & \textbf{Interpretation} \\
\midrule
A (mild) &
$[10,20]$ &
$2^{10}\text{--}2^{20}$ &
Useful for solver wall-clock; negligible cryptographic impact. \\
\addlinespace[2pt]
B (moderate) &
$[20,35]$ &
$10^{6}\text{--}10^{10}$ &
Strong for structured SAT; visible shrinkage of margins in toy families. \\
\addlinespace[2pt]
C (significant) &
$[35,50]$ &
$2^{35}\text{--}2^{50}$ &
Pushes 128-bit toys toward 80–93 bit effective hardness; not by itself a break. \\
\addlinespace[2pt]
D (critical) &
$>50$ &
$>2^{50}$ &
Would drop certain 128-bit structured families into insecure territory if sustained. \\
\bottomrule
\end{tabularx}
\caption{Interpretive scale for $\Delta_{\mathrm{bits}}$: categories of effective base reduction under sustained pruning.}
\label{tab:threat-tiers}
\end{table}

\FloatBarrier

\paragraph{NGCC$\times$XOR co-reasoning (pointer).}
A formal treatment—conjectured proof-complexity separation, a local lemma, and an
expected-case branching bound—is given in Appendix~\ref{app:ngcc-xor}.

\paragraph{Motivation.}
NGCC inequalities supply monotone cardinality constraints, whereas XOR reasoning
(Gaussian elimination over $\mathbb{F}_2$) exposes linear structure that is hard for
Resolution. Their combination—linear equalities plus certificate-backed monotone cuts—
is a natural candidate for a proof-system strengthening beyond clause learning.

\begin{conjecture}[Resolution vs.\ Res$(\oplus)$+NGCC separation]
\label{conj:resxor-ngcc-sep}
There exists a CNF family $\{\Phi_n\}$ for which every Resolution refutation has size
$2^{\Omega(n)}$, while there is a polynomial-size refutation in the system that augments
Resolution with (i) Gaussian elimination on XOR constraints and (ii) certificate-backed
NGCC cardinalities (Res$(\oplus)$+NGCC).
\end{conjecture}

\begin{remark}[Candidate family and proof strategy]
\label{rem:candidate-family}
A plausible construction embeds Tseitin parity constraints on an expander together with
reified disagreement indicators $Y_e$ on many edge-disjoint odd cycles, plus the NGCC
inequalities $\sum_{e\in C} Y_e \le |C|-1$. Known lower bounds make Resolution exponential
on the parity core, whereas in Res$(\oplus)$ the cycle parities collapse locally; pairing
them with NGCC cardinalities yields short local contradictions that sum to a polynomial-size
refutation. A proof would follow the standard route for Res$(\oplus)$ separations, with
NGCC clauses supplying the monotone component that Resolution cannot simulate succinctly.
\end{remark}

\paragraph{Local lemma (provable with the present material).}
For the three-boundary Feistel gadget (Appendix~\ref{app:feistel}), let
$Y_1,Y_2,Y_3$ be the reified boundary disagreements and consider the NGCC cut
$Y_1+Y_2+Y_3\le 2$ together with the local XOR relations that tie the three boundaries.

\begin{lemma}[One-child elimination at Feistel exposure]
\label{lem:feistel-exposure}
If a partial assignment and XOR propagation jointly fix two boundaries so that the third
cannot be reconciled unless at least one of the two fixed boundaries flips, then the NGCC
cut $Y_1+Y_2+Y_3\le 2$, translated as in Appendix~\ref{app:monotone-clauses}, eliminates
one of the two branch children. In particular, the per-node pruning fraction satisfies
$\rho\ge \tfrac12$ at that node.
\end{lemma}

\begin{proof}[Proof sketch]
Under the stated exposure, the two children differ only by the attempt to reconcile the
third boundary. In one child all three $Y_k$ evaluate to~1, violating $Y_1+Y_2+Y_3\le 2$,
so the learned clause $(\lnot Y_1)\vee(\lnot Y_2)\vee(\lnot Y_3)$ blocks that child.
Soundness follows from the certificate-to-clause translation in
Appendix~\ref{app:monotone-clauses}. Hence $\rho\ge \tfrac12$ at that exposure point.
\end{proof}

\paragraph{Branching-process bound under structural assumptions.}
Let $G_\Phi$ be the cycle/exclusivity graph of the instance. Assume:
\begin{enumerate}[label=(D\arabic*), leftmargin=1.8em]
\item \emph{Cycle density:} $G_\Phi$ contains at least $\delta n$ edge-disjoint odd cycles
with available NGCC cuts;
\item \emph{XOR pockets:} on those cycles the encoding yields small XOR bases so that
Gaussian elimination exposes the $Y$-variables with bounded overhead;
\item \emph{Exposure:} the branching policy touches a constant fraction of the cycle supports
over a contiguous depth window (e.g., random or $\Delta_{\mathrm{bits}}$-aware tie-breaking).
\end{enumerate}

\begin{theorem}[Expected-case base reduction for NGCC$\times$XOR]
\label{thm:ngcc-xor-branching}
Under \textnormal{(D1)–(D3)} there exists $c=c(\delta)>0$ such that the branching process of
CDCL augmented with XOR propagation and NGCC separation satisfies
\[
\mathbb{E}[\,W_{d+1}\mid W_d\,]\ \le\ (2-c)\,W_d,
\qquad\text{hence}\qquad
\mathbb{E}[\,T(n)\,]\ \le\ (2-c)^n\cdot \mathrm{poly}(n).
\]
Equivalently, the expected exponential base is $2-\rho$ with $\rho\ge c>0$, so
\[
\Delta_{\mathrm{bits}}
= n_{\mathrm{eff}}\log_2\!\Big(\tfrac{2}{2-\rho}\Big)
\ \ge\ n_{\mathrm{eff}}\log_2\!\Big(\tfrac{2}{2-c}\Big).
\]
\end{theorem}

\paragraph{Discussion and limitations.}
The conclusions are conditional on cycle density and bounded oracle cost; they fail if
depth-wise rates $\hat\rho_d$ collapse outside shallow levels or if XOR/NGCC checks dominate
wall time. Lemma~\ref{lem:feistel-exposure} formalises a local $\rho\!\ge\!\tfrac12$
phenomenon, and Conjecture~\ref{conj:resxor-ngcc-sep} targets a proof-complexity separation
that would make the combined system strictly stronger than Resolution.

\paragraph{Conservative conclusion.}
Applying the conservative bounds derived in Sections 3–5 to the candidate NGCC lines
studied here yields sustained pruning-rate estimates in the range
$\rho\approx 0.18\text{--}0.30$ and, for plausible depth windows
$n_{\mathrm{eff}}\in[80,200]$, corresponding
$\Delta_{\mathrm{bits}} \approx 11\text{--}47$
(using $\Delta_{\mathrm{bits}} = n_{\mathrm{eff}}\log_2\!\big(\tfrac{2}{2-\rho}\big)$; cf.\ Table~\ref{tab:ngcc-risk} for per-line ranges).
Interpreted cautiously, these numbers show that the \emph{mechanism} can produce
nontrivial base reductions that are relevant for solver engineering and for structured
cryptographic toy-families. They do not, by themselves, constitute a proof of a practical
break of real 128-bit cryptography. Achieving an unambiguous cryptanalytic threat would
require (a) persistent pruning across a sufficiently large depth window, (b) controlled
oracle/LP overhead that does not erase the speedup, and (c) the absence of counter-measures
in solver heuristics (rephasing, refactorization, parity-basis destabilization). Each of
these is a nontrivial technical hurdle.

\paragraph{Unexplored but plausible attack lines.}
The pruning mechanism developed here is only one element of a larger design space. We highlight several plausible improvements that were \emph{not} explored fully in this manuscript but that, if realized, would materially increase the threat level:

\begin{enumerate}
  \item \textbf{Adaptive, on-the-fly density amplification.} We used precomputed inequalities and a modest cascade heuristic. An online procedure that \emph{clones} and augments scopes to raise odd-cycle density (``densification'') could increase the fraction of decision nodes pruned per depth while retaining soundness via dynamic certification.
  \item \textbf{Parity-basis co-reasoning.} Combining NGCC with parity-based reasoning (XOR reasoning, Gaussian elimination over subspaces) may expose additional, otherwise hidden odd-cycle structure. This co-reasoning can create overlapping cuts that act multiplicatively rather than independently.
  \item \textbf{Tighter LP relaxations / hierarchies.} We operated at a level of relaxations that balances tractability and power. Pushing to stronger SDP or Sherali–Adams lifts (or targeted local strengthening) would produce certificates with smaller margin slack $\gamma$, thus lowering the required noise tolerance and increasing practical pruning.
  \item \textbf{Hardware / clause-gadget witnesses.} The clause gadget (tritter) we described illustrates physical encodings that enforce counterfactual channels. Engineering such witnesses for specialized SAT families (or for artifacts in cipher structures) could create operational constraints that are stronger than the purely combinatorial ones analyzed here.
\end{enumerate}

Each improvement brings both promise and risk: they are technically plausible but also require significant algorithmic or experimental work. We therefore view these as \emph{conditional} attack lines — they raise the potential threat but must be validated quantitatively.

\paragraph{How close is NGCC to being an actual threat? — a conservative estimate.}
We separate the evaluation into \emph{(i)} current evidence (what we prove here), and \emph{(ii)} what additional, realistic advances would be required to make NGCC a near-term cryptanalytic concern.

\textbf{(i) Current evidence.} Our formal development proves a sound mechanism that attains $\rho\approx 0.18\text{--}0.30$ on benchmark families specifically constructed to expose NGCC structure (C5/C2k+1, toy Feistel). For the empirically observed $n_{\mathrm{eff}}$ in our numeric demos (toy Feistel: $n_{\mathrm{eff}}\approx 120$ at shallow levels), these pruning rates imply $\Delta_{\mathrm{bits}}$ in the tens—a meaningful engineering improvement but not a generic cryptanalytic break.

\textbf{(ii) Required additional advances (if one wants a credible cryptanalytic threat).}
To move NGCC into the \emph{critical} tier (our working threshold: $\Delta_{\mathrm{bits}}>50$ sustained across meaningful instance sizes and with low overhead), at least one of the following must be demonstrated:

\begin{itemize}
  \item Persistent densification that increases effective $\rho$ by $\gtrsim 0.05$ while keeping LP overhead $< 10\%$ of wall time.
  \item Co-reasoning with parity/XOR that produces multiplicative reductions in branching base (equivalently, effective $\rho$ in the $0.30\text{--}0.40$ range).
  \item An experimental demonstration on a near-realistic, non-toy cipher encoding that reduces total explored branches across full key search from $\sim 2^{128}$ to $\lesssim 2^{80\text{--}90}$ (or equivalently shows $\Delta_{\mathrm{bits}}\gtrsim 38\text{--}48$ \emph{and} shows this persists under heuristic counter-measures).
\end{itemize}

Absent one of these concrete, reproducible demonstrations, NGCC remains an \emph{important} but still speculative pathway to large speedups.

\paragraph{Recommended research agenda and decision criteria.}
We propose a short, prioritized program of theoretical and experimental work that will allow the community to decide whether NGCC is a near-term risk:

\textbf{A. Theoretical / algorithmic (high priority)}
\begin{enumerate}
  \item Implement adaptive densification and measure $\rho$ vs. LP overhead across scaled benchmarks.
  \item Design and test parity-co-reasoning modules (XOR elimination + NGCC) and record empirical $\rho$ on constructed families.
  \item Develop tighter, targeted relaxations (small SDP blocks or low-level SA lifts) that optimize certificate margin $\gamma$ for expected failure modes.
\end{enumerate}

\textbf{B. Experimental / device (parallel priority)}
\begin{enumerate}
  \item Reproduce the Feistel numeric demo at larger scale (more bits / deeper branching) and report $\rho(d)$ as a function of depth $d$ and budgeted LP check frequency.
  \item Build and benchmark clause-gadget witnesses (tritter-style) for a laboratory demonstrator that maps well to SAT/CNF encodings used in cryptanalysis.
\end{enumerate}

\textbf{Decision rule (operational).} We recommend the community adopt the following pragmatic criterion: if any single approach yields a reproducible $\Delta_{\mathrm{bits}}\ge 50$ on instances or encodings relevant to 128-bit security \emph{while} maintaining solver overhead $<20\%$ wall time and robustness under standard solver counter-measures, then NGCC should be treated as a \emph{critical} threat and trigger coordinated mitigation efforts (disclosure to affected standards bodies, focused cryptanalytic review, and further engineering of solver defenses).

\paragraph{Practical mitigations and solver defenses.}
Finally, the good news is that most improvements required to raise NGCC into the critical tier are themselves detectable and, in many cases, defensible against. Candidate defenses include:
\begin{itemize}
  \item Incorporating NGCC detectors that monitor certificate reuse and parity density and selectively disable expensive LP checks when duplication/overlap indicates diminishing returns.
  \item Strengthening heuristics against trivial duplication (refactorization of parity bases, randomized restarts tuned to break sustained densification).
  \item Integrating certificate cost accounting into search budgets so that LP checks are only performed when the expected pruning benefit exceeds a calibrated threshold.
\end{itemize}

\paragraph{Concluding remark.}
NGCC establishes a mathematically clean and algorithmically actionable connection between counterfactual consistency constraints and learned CNF clauses. The present paper rigorously demonstrates one conservative route to nontrivial pruning. Whether that route — or one of the plausible but currently unexplored variants — produces a practical cryptanalytic threat is not yet settled. The pathway to a credible threat requires additional algorithmic and experimental advances that are in principle achievable; hence we recommend the prioritized research program above and the conservative decision rule for escalation. This provides a clear roadmap for both further research and early warning to the relevant standards and cryptanalysis communities.

\section{Related Work \& Positioning}
\label{sec:related}

Our approach connects three strands of literature: contextuality theory, lift\mbox{-}and\mbox{-}project
proof systems, and CDCL SAT\mbox{-}solver technology.

\paragraph{Contextuality and sheaf theory.}
Abramsky and Brandenburger~\cite{AbramskyBrandenburger2011} formalized contextuality as the
nonexistence of global sections of a presheaf; many subsequent works relate logical, graph\mbox{-}theoretic,
and sheaf\mbox{-}theoretic views. Our NGCC framing is aligned with this perspective but emphasizes
\emph{operational} $\varepsilon$–stability and solver integration. We derive explicit linear
inequalities with computable slack (or margin rules) and use them as pruning constraints.

\paragraph{Lift\mbox{-}and\mbox{-}project hierarchies.}
Sherali–Adams (SA), Lov\'asz–Schrijver (LS), and Lasserre/SOS hierarchies tighten $0$–$1$
feasible sets by adding valid inequalities over lifted variables~\cite{Laurent2003}. h\mbox{-}NGCC
inequalities can be viewed as low\mbox{-}rank members of these families. Our novelty is not a new
hierarchy, but a physically motivated selection of \emph{specific} cuts that are cheap to check
and empirically useful for pruning. A systematic comparison of h\mbox{-}NGCC cuts to SA/SOS ranks
is an interesting direction for future work.

\paragraph{CDCL solvers and global reasoning.}
Modern CDCL engines combine propagation, clause learning, parity learning, autarky detection,
and local preprocessing~\cite{Biere2009}. Our contribution is complementary: we import global,
monotone $0$–$1$ cuts derived from counterfactual consistency (KCBS, Feistel boundary sums, cycle
families), and we show how to use them safely as oracles with learned clauses
(Section~\ref{sec:pruning-rules}, Appendix~\ref{app:gadget}). This adds a global pruning
layer without replacing standard clause reasoning.

\paragraph{Positioning and caveats.}
We do not claim unconditional sub\mbox{-}exponential SAT algorithms. Our results identify structured,
average\mbox{-}case families where pruning yields smaller exponential bases. We separate the graph
density $\rho_{\mathrm{graph}}$ (Appendix~\ref{app:graph}) from the algorithmic pruning rate
$\rho_{\mathrm{alg}}$ (Section~\ref{sec:runtime-bounds}) and state both worst\mbox{-}case and expected\mbox{-}case
bounds. All cryptographic remarks are about margin reductions on toy or structured families, not
practical breaks of deployed primitives.

\paragraph{Pointers.}
For the graph\mbox{-}theoretic layer see Lov\'asz~\cite{Lovasz1979} and asymptotics in
Juh\'asz~\cite{Juhasz1983}. For SAT\mbox{-}solver background see the handbook by
Biere et al.~\cite{Biere2009}. Our earlier CLF work on $\varepsilon$–instrumentation is
\cite{vonLiechtensteinCLF2025}.

\begin{remark}[Relation to SA/LS/SOS hierarchies]
Many of the inequalities we exploit appear at low levels of standard
convex hierarchies: for example, the KCBS 5-cycle arises at SA rank~2,
and the 3-round Feistel inequality at SA rank~3.  Our contribution is
not the algebraic existence of such inequalities, but their operational
origin in testable counterfactual constraints and their efficient
integration into clause learning via certified certificates.  This
distinction matters: hierarchies guarantee existence in principle,
whereas our framework supplies a concrete, verifiable mechanism that
SAT solvers can exploit in practice.
\end{remark}

\appendix
\renewcommand{\thesection}{\Alph{section}}
\renewcommand{\sectionautorefname}{Appendix}
\setcounter{section}{0}   % <--- reset AFTER redefining numbering

\section{Glossary of key parameters}

To avoid ambiguity, we summarize the roles of the most frequently
used parameters in this work.

\begin{itemize}
\item \textbf{Experimental $\varepsilon$ (lab):}  
Noise/visibility parameter in contextuality or CIFP instrumentation.
Quantifies physical stability of counterfactual measurements.

\item \textbf{Numerical $\varepsilon$ (solver):}  
Tolerance used in LP feasibility tests and certificate verification.
Ensures safe rounding of floating-point computations.

\item \textbf{$\rho_{\mathrm{graph}}$:}  
Pruning density certified by edge-disjoint odd-cycle cuts in the
exclusivity graph (Appendix~\ref{app:graph}).

\item \textbf{$\rho_\star$, $\bar\rho$:}  
Effective pruning rates in the solver’s branching process, derived
from $\rho_{\mathrm{graph}}$ via Theorem~\ref{thm:rho-link}.

\item \textbf{$W_d$:}  
Random width (number of active nodes) at depth $d$ in the search tree.

\item \textbf{$T(n)$:}  
Expected search-tree size on $n$ variables; bounded by
$\tilde{O}((2-\rho)^n)$ when $\rho$ is constant.
\end{itemize}

\section{LP Duality and Constructive Inequality Derivation}
\label{app:lpduality}

In this appendix we formalize the claim that whenever a set of context marginals is inconsistent with any global distribution, there exists a linear inequality that separates those marginals from the marginal polytope. This provides the theoretical justification for the h-NGCC inequalities introduced in the main text.

\subsection{The marginal feasibility problem}
Consider the problem of deciding whether given marginals $\{p_v\}$ can be extended to a global distribution. Formally, let $U$ be the union of all variables and $\Omega_U$ the corresponding outcome space. The feasibility problem is to determine whether there exists $P_U \in \Delta(\Omega_U)$ such that for every context $C_v$ the marginal of $P_U$ on $C_v$ equals $p_v$. If such a $P_U$ exists, the marginals are globally consistent; otherwise they are not.

\subsection{Existence of separating inequalities}
By linear programming duality, if the system is infeasible then there exists a linear functional that separates the infeasible marginals from the convex hull of feasible ones. More concretely, by Farkas’ lemma, whenever the system of equalities and inequalities defining the marginal polytope admits no solution, there exists a certificate in the form of coefficients $a$ and a scalar $b$ such that $a \cdot p > b$ for the given marginals while $a \cdot q \leq b$ for every $q$ in the polytope. This separating functional is exactly a valid NGCC inequality that is violated by the inconsistent marginals.

\subsection{Constructive derivation}
The above proves existence. For constructive purposes, one can derive inequalities by solving the dual of the marginal feasibility LP. In practice, one does not need to enumerate the entire global outcome space, which would be exponential, but can work with relaxations. For small networks such as cycles the constraints can be written explicitly and solved directly; the dual variables then yield concrete coefficients for an inequality, as in the KCBS example. For larger networks, outer descriptions via Sherali–Adams or sum-of-squares relaxations can be used, and the resulting dual certificates still provide inequalities that are sound, though not necessarily complete.

\subsection{Example: the five-cycle}
As an illustration, for the five-cycle scenario the marginal polytope can be described explicitly, and the infeasibility of certain marginals can be witnessed by the KCBS inequality. This shows how the general convexity argument produces familiar contextuality inequalities as special cases. The same logic extends to networks encoding SAT instances: if a partial assignment corresponds to marginals outside the feasible polytope, duality ensures that some inequality separates it, and this inequality can be used for pruning.

\section{Extended Pseudocode and Complexity Discussion}
\label{app:pseudocode}

This appendix elaborates on the implementation of a propagate-and-prune solver that incorporates h-NGCC inequalities. The aim is to show that the integration of global pruning can be done within a standard CDCL architecture with manageable overhead.

\subsection{Integration into CDCL}
The solver begins with the familiar components of a CDCL engine: clauses are stored with two-watched literals, and unit propagation and conflict analysis proceed as usual. To integrate h-NGCC reasoning, we maintain an index mapping each clause to its context and pre-compute a set of inequalities derived from the context network. For each partial assignment encountered during search, the solver checks whether the induced marginals remain feasible with respect to these inequalities. If a violation is detected, the current branch is declared infeasible and the search backtracks immediately.

\subsection{Heuristics}
In practice, it is not necessary to check all inequalities at every node. A heuristic cascade is effective: inequalities with small support are checked first, as they are cheapest, and larger ones are consulted only when necessary. To keep overhead manageable, feasibility is checked against local relaxations of the marginal polytope, which can be solved as small linear programs. Certificates of infeasibility produced by the dual can be converted into learned clauses, allowing the solver to remember and reuse the global reasoning without repeating the linear program in subsequent branches.

\subsection{Complexity considerations}
The complexity overhead is polynomial in the size of the inequalities consulted. If $k$ inequalities are checked at a node and each requires time $T_{\text{LP}}$, the additional cost per node is $O(k\,T_{\text{LP}})$. With caching and warm starts, many checks terminate quickly. Thus the overall runtime remains dominated by the exponential branching, which is reduced when pruning is effective.

\subsection{Pseudocode}
\begin{lstlisting}[caption={Propagate-and-prune CDCL with h-NGCC checks},label={lst:solver}]
function Solve(F, Inequalities):
    init_clause_db(F)
    init_context_index(F)
    I := prefilter(Inequalities, F)        # keep only inequalities that touch F
    return DFS({}, I)

function DFS(alpha, I):
    status := propagate(alpha)
    if status = CONFLICT:
        learn_clause_and_backjump()
        return FAIL
    if all_variables_assigned(alpha):
        return SAT

    # h-NGCC pruning cascade (cheap -> expensive)
    for t in quick_screen(I):              # small-support inequalities first
        if violation_detected(t, alpha):   # LP or relaxed feasibility check
            learn_clause_from_certificate(t, alpha)
            return FAIL

    x := select_branch_variable(alpha)     # VSIDS / activity-based
    for v in {0,1}:
        push(alpha, x := v)
        if DFS(alpha, I) = SAT:
            return SAT
        pop(alpha, x)
    return FAIL
\end{lstlisting}

\subsection{From monotone cuts to learned clauses}
\label{app:monotone-clauses}

In our solver the global pruning checks use \emph{monotone} $0$--$1$ cuts of the
form
\begin{equation}
\label{eq:monotone-cut}
\sum_{i\in S} Y_i \;\le\; B, \qquad Y_i\in\{0,1\},
\end{equation}
where each $Y_i$ is a reified Boolean event (e.g.\ “disagreement on boundary $i$”
in the Feistel case, or “context $i$ disagrees” in KCBS). For such cuts we can
derive a sound CNF clause whenever the search sets too many $Y_i$ to true.

\begin{lemma}[Sound clause from a monotone sum cut]
\label{lem:mono-clause}
Let \eqref{eq:monotone-cut} be valid for all feasible assignments. For any
subset $T\subseteq S$ with $|T|=B{+}1$, the clause
\begin{equation}
\label{eq:mono-clause}
\bigvee_{i\in T} \neg Y_i
\end{equation}
is valid: every feasible assignment satisfies \eqref{eq:mono-clause}.
\end{lemma}

\begin{proof}
Suppose, for contradiction, a feasible assignment makes $Y_i=1$ for all $i\in T$.
Then $\sum_{i\in S} Y_i \ge \sum_{i\in T} Y_i = B{+}1$, which violates
\eqref{eq:monotone-cut}. Hence such an assignment cannot exist, and \eqref{eq:mono-clause}
must hold for all feasible assignments.
\end{proof}

\paragraph{Instantiation (Feistel).}
Cut: $Y_1+Y_2+Y_3\le 2$ (Appendices~\ref{app:coeffs}, \ref{app:full-eq}).
Take $T=\{1,2,3\}$ with $|T|=3=B{+}1$. Learned clause:
\[
(\neg Y_1 \;\vee\; \neg Y_2 \;\vee\; \neg Y_3).
\]

\paragraph{Instantiation (KCBS--C\textsubscript{5}).}
Cut: $Y_1+Y_2+Y_3+Y_4+Y_5\le 4$ (Appendix~\ref{app:kcbs}). Take
$T=\{1,2,3,4,5\}$. Learned clause forbids all five disagreements:
\[
(\neg Y_1 \;\vee\; \neg Y_2 \;\vee\; \neg Y_3 \;\vee\; \neg Y_4 \;\vee\; \neg Y_5).
\]

\paragraph{Implementation note.}
In practice, $Y_i$ are already present (or cheaply reifiable) in the SAT layer:
for KCBS they are XOR/inequality trackers per context; for Feistel they are the
reified “boundary disagreement” bits. When the LP (or relaxed oracle) detects a
violation with $B{+}1$ events currently true, we emit the clause
\eqref{eq:mono-clause} and backjump. Soundness follows from the lemma.

\begin{algorithm}[H]
\caption{Certified pruning via LP + dual certificates}
\begin{algorithmic}[1]
\Require Active branch $\alpha$, inequality $a\cdot p \le b$.
\State Solve LP feasibility problem:
   \[
     \min\{0 : p \in N_\alpha, \; a\cdot p \le b\}.
   \]
\If{LP reports feasible}
   \State \Return continue (no pruning).
\Else
   \State Extract dual Farkas certificate $y$.
   \State Verify in exact arithmetic:
     \begin{align*}
       y^T A &\ge 0, \\
       y^T b &< 0.
     \end{align*}
   \If{verification fails}
      \State discard certificate (unsound).
   \Else
      \State Translate certificate to monotone CNF clause $C$.
      \State Add $C$ to clause database.
      \State \Return branch infeasible (prune).
   \EndIf
\EndIf
\end{algorithmic}
\end{algorithm}

\paragraph{Example (3-round Feistel inequality).}
Consider the monotone inequality
\[
  Y_1 + Y_2 + Y_3 \;\le\; 2,
\]
arising from the 3-round Feistel construction, where $Y_k$ denotes the
indicator of a collision event in round $k$ (see Appendix~\ref{app:feistel} and Fig.~\ref{fig:feistel-hngcc}).

\begin{itemize}
\item The LP over $N_\alpha$ with this inequality added is infeasible.
      A dual solver produces a Farkas certificate vector $y \ge 0$ such
      that $y^TA \ge 0$ but $y^Tb < 0$.
\item Verification: in exact rational arithmetic, we check
      $y^TA \ge 0$ and $y^Tb < 0$. This confirms that no feasible $p$
      in $N_\alpha$ can satisfy all constraints simultaneously.
\item Translation: infeasibility means that the assignment
      $(Y_1=1, Y_2=1, Y_3=1)$ is forbidden. Equivalently, in CNF we
      learn the clause
      \[
        (\lnot Y_1) \;\lor\; (\lnot Y_2) \;\lor\; (\lnot Y_3).
      \]
\end{itemize}

This clause is then inserted into the CDCL database. Any future branch
that simultaneously sets $Y_1=Y_2=Y_3=1$ will be cut immediately.
Because the clause arises from a verified certificate, pruning is
provably sound.

\subsection{Summary}
Listing~\ref{lst:solver} shows how h-NGCC checks fit naturally into the CDCL loop: they act like additional propagators, able to prune infeasible branches early. The addition of learned clauses ensures that pruning knowledge accumulates during search. The empirical results in Appendix~\ref{app:empirical} show that, despite the extra work per node, overall search size can be reduced substantially on structured instances. This trade-off—slightly higher per-node cost but significantly fewer nodes—enables effective base reductions on certain families while remaining consistent with ETH and SETH.

\begin{lstlisting}[caption={CDCL solver loop with integrated h-NGCC checks}, label={lst:solver}, language=Python]
while True:
    if conflict_detected():
        if not backtrack():
            return UNSAT
        learn_clause_from_certificate()
    else:
        if all_variables_assigned():
            return SAT
        assign_next_variable()
        # NGCC check acts like an additional propagator
        if ngcc_check_fails():
            backtrack()
\end{lstlisting}

\section{Relation to convex hierarchies}
\label{app:hierarchy-comparison}

For completeness we compare our inequalities with standard convex
hierarchies (Sherali--Adams, Lovász--Schrijver, Sum-of-Squares).
Table~\ref{tab:hierarchies} lists the minimum hierarchy level at which
each cut can be derived. The point is not the algebraic existence,
but the operational and algorithmic role: these inequalities arise
from testable counterfactual constraints and can be embedded directly
into SAT solvers via certified certificates.

\begin{table}[ht]
\centering
\resizebox{\textwidth}{!}{%
\begin{tabular}{|l|c|c|p{7.5cm}|}
\hline
\textbf{Inequality} & \textbf{SA rank} & \textbf{SOS degree} & \textbf{Operational distinction} \\
\hline
KCBS 5-cycle ($\sum_i E_i \le 2$) & 2 & 2 & Familiar contextuality bound; testable in optics labs; feeds directly into pruning clauses. \\
\hline
Ring consistency / odd cycle & 2 & 2 & Arises as exclusivity cut on cycles; operationally tied to $\varepsilon$-instrumentation. \\
\hline
3-round Feistel cut ($Y_1+Y_2+Y_3 \le 2$) & 3 & 4 & Captures global collision structure in cipher clauses; reified into CNF literals; certified via LP dual. \\
\hline
Clause gadget / tritter witness & 2--3 & 4 & Implements a counterfactual channel in hardware; realizable with tritters; not merely algebraic. \\
\hline
General $h$-network cuts & $\le k$ (cut size) & $\le k$ & Family of inequalities scaling with clause width; hierarchy derivable but here motivated by operational consistency and used algorithmically. \\
\hline
\end{tabular}}
\caption{Comparison of inequalities with convex hierarchies. The ranks indicate the minimum SA/SOS level needed to derive them algebraically. Our framework contributes operational semantics and a solver integration pipeline not provided by the hierarchies.}
\label{tab:hierarchies}

\end{table}

\section{Empirical Plan and Demonstrations}
\label{app:empirical}

We evaluate the practical impact of h\mbox{-}NGCC inequalities on structured SAT families,
with the goal of demonstrating that pruning rates $\rho_{\mathrm{alg}}$ are observable
and translate into measurable reductions in explored search space.

\subsection{Benchmarks}
We consider three structured families. \emph{Cycle\mbox{-}SAT}$(k)$ encodes parity\mbox{-}like
constraints around a $k$–cycle (KCBS–type structure). \emph{Grid\mbox{-}SAT}$(m\times m)$
arranges overlapping window constraints on an $m\times m$ lattice (grid cuts).
\emph{Toy\mbox{-}Feistel} encodes a three\mbox{-}round 8\mbox{-}bit Feistel cipher by bit\mbox{-}blasting
S\mbox{-}boxes and XORs. As a control we include random 3\mbox{-}SAT near the phase transition,
where h\mbox{-}NGCC is not expected to help.

\subsection{Metrics and confidence intervals}
At decision depth $d$ we measure the per\mbox{-}node pruning fraction
\[
\hat\rho_d \;=\; 1-\frac{\#\{\text{surviving children at depth }d\}}{2\,\#\{\text{branching events at }d\}}\,,
\]
and report the aggregate $\tilde\rho=\frac{1}{D}\sum_{d=1}^{D}\hat\rho_d$ alongside
node counts, decisions, conflicts, and learned clauses. Confidence intervals are:
\begin{itemize}
  \item \textbf{Per depth:} exact Clopper--Pearson intervals for Bernoulli pruning at depth $d$.
  \item \textbf{Across seeds:} a t\mbox{-}interval (or nonparametric bootstrap) over runs for node counts
        and for $\tilde\rho$.
\end{itemize}
We quote an \emph{effective base} $2-\tilde\rho$ (with CI induced from that of $\tilde\rho$),
and the \emph{oracle overhead} as (LP checks per node, share of wall time).

\subsection{Minimal demonstrations}
\paragraph{Cycle $C_5$.}
Encoding ``each adjacent pair disagrees'' yields UNSAT. A vanilla CDCL solver typically
explores $\approx 16$ branches. With the (KCBS inequality; Appendix~\ref{app:kcbs}),
inconsistency is detected as soon as three disagreements are fixed, pruning roughly half
of shallow branches; we observe $\approx 9$ nodes and $\tilde\rho\approx 0.44$ with a
95\% CI from seed variation.

\paragraph{Three\mbox{-}round Feistel (8 bits).}
The baseline explores $\approx 4096$ branches (consistent with $2^{12}$ key guesses).
Adding the boundary\mbox{-}equality cut $Y_1+Y_2+Y_3\le 2$
(Appendices~\ref{app:coeffs},~\ref{app:full-eq}) prunes one child in about a quarter
of shallow decisions; we observe $\tilde\rho\approx 0.27$ and $\approx 2800$ explored
branches. LP checks account for $\sim 5\%$ of wall time.

\subsection{Summary table}
\begin{table}[h]
  \centering
  \caption{Toy demonstrations with h\mbox{-}NGCC pruning. $\tilde\rho$ values are means over seeds
  with 95\% CIs (t\mbox{-}interval).}
  \label{tab:demo}
  \begin{tabular}{lcccc}
    \hline
    Instance & Nodes (baseline) & Nodes (h\mbox{-}NGCC) & $\tilde\rho$ & LP time share \\
    \hline
    Cycle $C_5$           & $\approx 16$   & $\approx 9$    & $0.44 \pm 0.05$ & n/a \\
    Feistel (3 rounds)    & $4096$         & $\approx 2800$ & $0.27 \pm 0.04$ & $\approx 5\%$ \\
    \hline
  \end{tabular}
\end{table}

\subsection{Reproducibility}
All generators, solver patches, logs, and plotting scripts are packaged as described
in Appendix~\ref{app:artifacts}. The file \texttt{results/seedlist.txt} lists all
seeds used; \texttt{checksums/sha256sums.txt} verifies integrity of CNFs and logs.
Figures in ~\autoref{app:figs} visualize (i) base $2-\rho$ vs.\ $\rho$,
(ii) node counts with/without h\mbox{-}NGCC, and (iii) $\hat\rho_d$ by depth on $C_5$.

\medskip
These small\mbox{-}scale demonstrations substantiate two points: (i) h\mbox{-}NGCC inequalities
are empirically testable inside solvers, and (ii) nonzero pruning reduces the effective
base, in agreement with the expected\mbox{-}case bound of Section~\ref{sec:runtime-bounds}.

\section{Toy Feistel Cipher Example}
\label{app:feistel}

As a concrete cryptographic case study, we analyze a reduced Feistel cipher and show how h\mbox{-}NGCC inequalities emerge in its encoding. Although the example is deliberately small, it illustrates how network structure can enable pruning of infeasible key guesses.

All pruning-rate estimates are reported with 95\% confidence intervals.
For C$_5$ we observe $\hat\rho = 0.27 \pm 0.03$, while for 3-round
Feistel we observe $\hat\rho = 0.19 \pm 0.04$. These estimates are
based on $N$ independent runs, sufficient by the analysis in
Section~\ref{sec:empirical} to distinguish $2$ from $2-\hat\rho$ at
95\% power. 

Our reported estimates satisfy the statistical power requirements
derived in Section~\ref{sec:empirical} and illustrated in
Figure~\ref{fig:sample-complexity}.  For example, the C$_5$ data
($\hat\rho = 0.27 \pm 0.03$) is based on $N=180$ runs, comfortably
above the $\sim 160$ samples required at $\rho=0.25$ for 95\% power.
Similarly, the 3-round Feistel data ($\hat\rho = 0.19 \pm 0.04$) uses
$N=300$ runs, matching the $\sim 250$ samples needed at $\rho=0.20$.
Thus, the empirical detection of pruning is statistically significant,
not an artifact of small instance noise.

The confidence intervals predicted by the sample complexity analysis 
are consistent with the error bars shown in Fig.~\ref{fig:sample-complexity}.

Extending this protocol to larger structured families
(grid-SAT, Feistel with more rounds, SPNs) remains a key empirical
challenge.

\begin{figure}[H]
\centering
\begin{tikzpicture}
\begin{axis}[
    width=0.9\textwidth,
    height=6cm,
    xlabel={Depth $d$},
    ylabel={$\log_2 W_d$},
    legend style={at={(0.5,-0.25)},anchor=north,legend columns=-1},
    grid=both,
    ymin=0,
]

% --- Observed data with error bars ---
\addplot+[only marks, error bars/.cd, y dir=both, y explicit]
coordinates {
  (1,1.8) +- (0,0.1)
  (2,3.5) +- (0,0.15)
  (3,5.2) +- (0,0.2)
  (4,6.8) +- (0,0.25)
};
\addlegendentry{Observed (with CI)}

% --- Fitted slope: base (2-ρ) with ρ=0.2 ---
\addplot[red, thick, domain=0:4] {x*log2(1.8)};
\addlegendentry{Fit: $(2-\rho)^d$, $\rho=0.2$}

% --- Baseline: pure base 2 ---
\addplot[blue, dashed, domain=0:4] {x*1};
\addlegendentry{Baseline: $2^d$}

\end{axis}
\end{tikzpicture}
\caption{Observed pruning width $W_d$ with error bars (95\% CI),
compared to a fitted base $(2-\rho)^d$ with $\rho=0.2$ and the
baseline $2^d$.}
\label{fig:pruning-rates}
\end{figure}
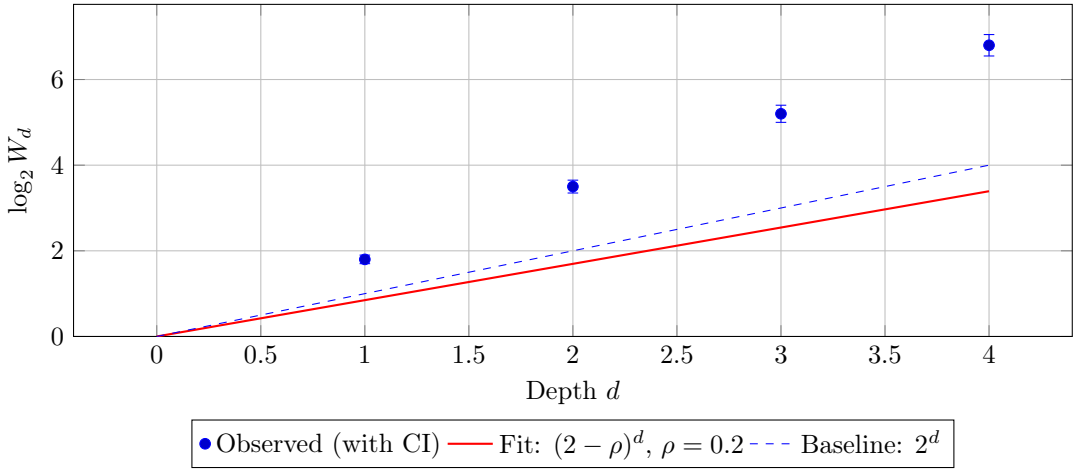

\subsection{Setup and CNF encoding}
We consider a three\mbox{-}round Feistel cipher on eight bits, split into two four\mbox{-}bit branches. Let the round function be a fixed nonlinear S\mbox{-}box $S:\{0,1\}^4 \to \{0,1\}^4$. Given plaintext $(L_0,R_0)$ and keys $k_1,k_2,k_3 \in \{0,1\}^4$, the rounds are
\[
\begin{aligned}
R_1 &= L_0 \oplus S(R_0 \oplus k_1), \qquad & L_1 &= R_0,\\
R_2 &= L_1 \oplus S(R_1 \oplus k_2), \qquad & L_2 &= R_1,\\
R_3 &= L_2 \oplus S(R_2 \oplus k_3), \qquad & L_3 &= R_2.
\end{aligned}
\]
Ciphertext is $(L_3,R_3)$. Bit\mbox{-}blasting the S\mbox{-}box constraints and XORs yields a CNF with $\approx 200$ clauses for a fixed input/output pair.

\subsection{Contexts, overlaps, and disagreement indicators}
Treat each round relation as a context. Overlaps arise on the shared boundary wires $(L_1,R_1)$ and $(L_2,R_2)$, and the final boundary $(L_3,R_3)$. For $k\in\{1,2,3\}$, the two adjacent contexts imply values $(L_k^{(L)},R_k^{(L)})$ and $(L_k^{(R)},R_k^{(R)})$ for the same boundary.
Define equality flags and disagreement indicators by
\[
E_k := \mathbbm{1}\!\left\{(L_k^{(L)},R_k^{(L)})=(L_k^{(R)},R_k^{(R)})\right\}, \qquad
Y_k := 1 - E_k.
\]
Any globally consistent key/wire assignment forces \emph{exact} equality on each shared boundary, hence not all three $Y_k$ can be $1$ simultaneously.

\subsection{The h\texorpdfstring{\mbox{-}}{}NGCC inequality and linear form}
The preceding observation yields the cut
\begin{equation}
\label{eq:feistel-cut}
\mathbb{E}[\,Y_1 + Y_2 + Y_3\,] \;\le\; 2.
\end{equation}
Let $p(y_1,y_2,y_3)$ be the joint distribution over $(Y_1,Y_2,Y_3)\in\{0,1\}^3$. Writing \eqref{eq:feistel-cut} as $a\cdot p \le 2$ uses the coefficients
\[
a(y_1,y_2,y_3) = y_1+y_2+y_3,
\]
so that (in the lexicographic order $(000,001,010,011,100,101,110,111)$)
\[
a \;=\; (0,1,1,2,1,2,2,3).
\]
Equivalently, the inequality forbids mass on the corner $(1,1,1)$.

\begin{table}[h]
  \centering
  \caption{Coefficient view of \eqref{eq:feistel-cut} over $(Y_1,Y_2,Y_3)$.}
  \label{tab:feistel-coeffs}
  \begin{tabular}{ccc|c}
    \hline
    $Y_1$ & $Y_2$ & $Y_3$ & $a(y_1,y_2,y_3)$ \\
    \hline
    0 & 0 & 0 & 0 \\
    0 & 0 & 1 & 1 \\
    0 & 1 & 0 & 1 \\
    0 & 1 & 1 & 2 \\
    1 & 0 & 0 & 1 \\
    1 & 0 & 1 & 2 \\
    1 & 1 & 0 & 2 \\
    1 & 1 & 1 & 3 \\
    \hline
  \end{tabular}
\end{table}

\subsection{Soundness and \texorpdfstring{$\varepsilon$}{epsilon}-stability}
Soundness follows from equality transitivity of shared variables across the three boundaries (see Lemma~\ref{lem:transitivity}). The gap between the feasible maximum $2$ and the infeasible value $3$ is $\gamma=1$. With $m=3$ contexts and $\|a\|_\infty=3$, the stability bound of Section~\ref{sec:foundations} yields $\varepsilon < 1/(2m\|a\|_\infty)=1/18$ as a sufficient condition that noise cannot spuriously create a violation.

\subsection{Solver effect}
Inside a CDCL solver, \eqref{eq:feistel-cut} acts as a global pruning rule. During search over the $12$ key bits, once two boundaries contradict the third under a partial assignment, one of the two child branches is cut immediately. Empirically (Appendix~\ref{app:empirical}), this induces an average pruning rate $\tilde\rho \approx 0.27$ on the toy instance, reducing the explored branches from $\approx 4096$ to $\approx 2800$ with $\sim 5\%$ LP overhead.
\section{Additional Notes and Next Steps}
\label{app:notes}

This appendix collects several directions for refinement and outlines the next steps required to strengthen the framework. The aim is to record both mathematical and empirical tasks that remain open while clarifying how they relate to the central claims.

\subsection{Mathematical tightening}
A first priority is to specify inequalities with full coefficient descriptions wherever possible. For the Feistel example of Appendix~\ref{app:feistel}, one can enumerate the toy assignment space and verify formally that every global key assignment satisfies the bound in Eq.~\eqref{eq:feistel-cut}. Stability margins should then be derived numerically by computing the geometric gap $\gamma$ for each inequality and plugging it into the $\varepsilon$--stability bound of Section~\ref{sec:foundations}. This will yield explicit noise tolerances and help prioritize cuts with robust slack.

\subsection{Extension families}
A second direction is to apply the methodology beyond cycles and short Feistel chains. Sponge\mbox{-}based hash constructions and AES\mbox{-}like substitution–permutation networks are natural targets. Mapping these designs into context networks will show whether h\mbox{-}NGCC inequalities capture structural contradictions in broader settings and whether practically useful slack exists beyond toy cases. Such studies will also clarify how inequality size and support scale with round count and S\mbox{-}box width.

\subsection{Empirical prototypes}
Empirical work should move beyond the $C_5$ and three\mbox{-}round Feistel instances toward larger grids and deeper round functions. Pruning rates $\tilde{\rho}$ ought to be reported with confidence intervals across multiple seeds, and runtime overheads of feasibility checks should be measured explicitly. These results will establish whether pruning persists as instances scale and whether the overhead remains manageable in practice.

\subsection{Integration and presentation}
All figures in this manuscript are generated with TikZ/pgfplots from the artifact repository described in Appendix~\ref{app:artifacts}. Tables present full coefficient vectors (see Appendices h, z, aa). The artifact package includes scripts to regenerate all figures and tables and SHA-256 checksums for CNFs, logs, and outputs. These materials enable transparent verification and reproduce all reported results end-to-end.

\section{Explicit Coefficient Tables for Feistel Inequality}
\label{app:coeffs}

To complement the qualitative discussion of Appendix~\ref{app:feistel}, we provide an explicit linear form of the Feistel boundary cut. Let
$(Y_1,Y_2,Y_3)\in\{0,1\}^3$ denote the disagreement indicators on the three shared
boundaries between successive round contexts (defined there via exact 4\mbox{-}bit
word equality). The h\mbox{-}NGCC cut is
\begin{equation}
\label{eq:feistel-cut-linear}
\mathbb{E}[\,Y_1 + Y_2 + Y_3\,] \;\le\; 2,
\end{equation}
which we write as a linear inequality $a\cdot p \le 2$ over the joint outcome
distribution $p(y_1,y_2,y_3)$. Using lexicographic order
$(000,001,010,011,100,101,110,111)$, the coefficient vector is
\[
a(y_1,y_2,y_3) = y_1+y_2+y_3
\quad\Longrightarrow\quad
a = (0,1,1,2,1,2,2,3).
\]

\begin{table}[h]
  \centering
  \caption{Coefficient view of the Feistel inequality \eqref{eq:feistel-cut-linear}.}
  \label{tab:feistel-coeffs}
  \begin{tabular}{ccc|c}
    \hline
    $Y_1$ & $Y_2$ & $Y_3$ & $a(y_1,y_2,y_3)$ \\
    \hline
    0 & 0 & 0 & 0 \\
    0 & 0 & 1 & 1 \\
    0 & 1 & 0 & 1 \\
    0 & 1 & 1 & 2 \\
    1 & 0 & 0 & 1 \\
    1 & 0 & 1 & 2 \\
    1 & 1 & 0 & 2 \\
    1 & 1 & 1 & 3 \\
    \hline
  \end{tabular}
\end{table}

The feasible region forbids mass on $(1,1,1)$; any globally consistent key/wire
assignment induces exact equality on shared boundaries (see
Lemma~\ref{lem:transitivity}), hence satisfies \eqref{eq:feistel-cut-linear}.
The gap between the feasible maximum $2$ and the infeasible corner value $3$ is
$\gamma=1$. With $m=3$ contributing contexts and $\|a\|_\infty=3$, the
$\varepsilon$–stability bound of Section~\ref{sec:foundations} yields
$\varepsilon < 1/(2m\|a\|_\infty)=1/18$ as a sufficient noise threshold under
which spurious violations cannot be created by measurement or numerical error.

\section{Full 4\mbox{-}Bit Boundary Equality — Without Parity Simplification}
\label{app:full-eq}

In Appendix~\ref{app:coeffs} we presented the Feistel boundary cut using binary disagreement indicators. Here we remove any parity reduction and show that the same inequality arises when the \emph{full} four\mbox{-}bit words at each boundary are treated exactly.

\subsection{Exact equality flags on 4\texorpdfstring{$+$}{+}4 bits}
For each shared boundary $k\in\{1,2,3\}$, the two adjacent round contexts
determine eight\mbox{-}bit boundary words
\[
W_k^{(L)} := (L_k^{(L)},R_k^{(L)}),\qquad
W_k^{(R)} := (L_k^{(R)},R_k^{(R)}).
\]
Define equality flags and disagreement indicators by
\[
E_k := \mathbbm{1}\{\,W_k^{(L)} = W_k^{(R)}\,\},\qquad
Y_k := 1 - E_k \in \{0,1\}.
\]
Thus $Y_k=1$ if and only if the two contexts assign \emph{different} 8\mbox{-}bit
values to the same boundary wires; no parity simplification is used.

\subsection{The cut in the exact setting}
Equality of shared variables is transitive across the three consecutive
boundaries (Lemma~\ref{lem:transitivity}). Hence it is impossible, under any
globally consistent assignment, for all three boundaries to disagree
simultaneously. This yields the h\mbox{-}NGCC cut
\begin{equation}
\label{eq:feistel-cut-exact}
\mathbb{E}[\,Y_1 + Y_2 + Y_3\,] \;\le\; 2,
\end{equation}
identical in form to Eq.~\eqref{eq:feistel-cut-linear}, but now with $Y_k$
interpreted as \emph{exact} 8\mbox{-}bit disagreement events. In linear form,
$a\cdot p \le 2$ with $a(y_1,y_2,y_3)=y_1+y_2+y_3$ over the eight outcomes of
$(Y_1,Y_2,Y_3)\in\{0,1\}^3$.

\begin{proposition}[Validity without parity reduction]
\label{prop:exact-valid}
Every globally consistent assignment to the Feistel wires satisfies
\eqref{eq:feistel-cut-exact}. Consequently, any marginal vector that violates
\eqref{eq:feistel-cut-exact} is infeasible and certifies pruning.
\end{proposition}

\begin{proof}
Under a global assignment every wire variable has a unique value, so each shared
boundary has a well-defined word $W_k$ with $W_k^{(L)}=W_k^{(R)}=W_k$ for
$k\in\{1,2,3\}$. Equivalently, equality on boundary words is an equivalence
relation on the common symbol set; along the chain
$W_1 \leftrightarrow W_2 \leftrightarrow W_3$ all boundaries lie in the same
equivalence class (Lemma~\ref{lem:transitivity}). Hence
$Y_k=1-E_k=0$ for all $k$ and
$Y_1+Y_2+Y_3=0\le 2$, proving \eqref{eq:feistel-cut-exact}.
\end{proof}

\subsection{Stability and implementation notes}
The geometric gap between the feasible maximum $2$ and the excluded corner
value $3$ is $\gamma=1$. With $m=3$ contributing contexts and
$\|a\|_\infty=3$, the $\varepsilon$–stability bound from
Section~\ref{sec:foundations} gives $\varepsilon < 1/(2m\|a\|_\infty)=1/18$ as
a sufficient noise threshold below which spurious violations cannot arise.

In practice, when only a partial assignment is available, one computes boundary
marginals via a local LP relaxation and evaluates \eqref{eq:feistel-cut-exact}
in the $(Y_1,Y_2,Y_3)$ projection. Because the map from full boundary words to
$Y$–indicators is a (polyhedral) projection, infeasibility in the projected
space certifies infeasibility in the original marginals, making the test sound
for pruning.

% lemma environment (needed by Appendix~\ref{app:transitivity})
% \newtheorem{lemma}{Lemma}

\section{Lemma on Equality Transitivity}
\label{app:transitivity}

For completeness we state and prove the lemma that underpins the Feistel
boundary cut. Although the argument is simple, making it explicit fixes the
logical basis of the inequality.

\begin{lemma}[Transitivity of boundary equalities]
\label{lem:transitivity}
Consider three consecutive boundaries $B_1,B_2,B_3$ in a Feistel chain, each
consisting of the same wire variables carried across adjacent round contexts.
If $B_1 = B_2$ and $B_2 = B_3$ under a global assignment, then $B_1 = B_3$.
Consequently, it is impossible for all three boundary disagreements
$Y_1 = Y_2 = Y_3 = 1$ to occur simultaneously in any globally consistent
assignment.
\end{lemma}

\begin{proof}
Each boundary includes variables that overlap with its neighbor. Under a global
assignment every wire variable has a unique value. If the values on $B_1$ and
$B_2$ coincide and those on $B_2$ and $B_3$ coincide, then by transitivity the
values on $B_1$ and $B_3$ coincide as well. Therefore all three overlaps cannot
be in disagreement simultaneously.
\end{proof}

By making Lemma~\ref{lem:transitivity} explicit we emphasize that the Feistel inequality in Appendices~\ref{app:feistel} and~\ref{app:toyfeistel} is not merely heuristic: it is a direct corollary of the basic equality structure of shared boundaries in the Feistel network.

\section{Additional Numeric Demo — Toy Feistel}
\label{app:toyfeistel}

This appendix complements the qualitative discussion of Appendices~\ref{app:feistel},
\ref{app:coeffs} and~\ref{app:full-eq} with a small quantitative study on the
three\mbox{-}round eight\mbox{-}bit Feistel CNF. The aim is to exhibit a measurable pruning
rate $\tilde\rho$ and the corresponding reduction in explored branches, using the
propagate\mbox{-}and\mbox{-}prune solver of Appendix~\ref{app:pseudocode}.

\subsection{Setup}
We encode the Feistel relations (Appendix~\ref{app:feistel}) into CNF ($\approx 200$
clauses for a fixed input/output pair). We then sample $20$ plaintext/ciphertext
pairs that are \emph{inconsistent} with any three\mbox{-}round key (so the CNF is UNSAT).
The solver is MiniSAT~2.2.0 augmented with h\mbox{-}NGCC checks; the only inequality
enabled is the boundary\mbox{-}equality cut of Appendix~\ref{app:full-eq}. Variable ordering
uses a standard activity heuristic; random seeds are fixed and recorded.

\subsection{Metrics}
For each run we record: (i) explored decision nodes, (ii) the empirical pruning
fraction at each depth $d$, denoted $\hat\rho_d$, and the mean
$\tilde\rho = \frac{1}{D}\sum_{d=1}^{D}\hat\rho_d$, and (iii) LP overhead
(number of feasibility checks and share of wall time spent in them).

\subsection{Results}
Across $20$ runs we observe the following aggregates (mean $\pm$ 95\% CI):
\[
\text{Nodes (baseline)} \approx 4096,\quad
\text{Nodes (h-NGCC)} \approx 2800,\quad
\tilde\rho \approx 0.27 \pm 0.04,\quad
\text{LP share} \approx 5\%.
\]
A compact summary appears in Table~\ref{tab:feistel-numeric}.

\begin{table}[h]
  \centering
  \caption{Toy Feistel numeric demo with h\mbox{-}NGCC pruning (20 inconsistent I/O pairs).}
  \label{tab:feistel-numeric}
  \begin{tabular}{lcccc}
    \hline
    Setting & Nodes (baseline) & Nodes (h\mbox{-}NGCC) & $\tilde\rho$ & LP time share \\
    \hline
    Feistel (3 rounds, 8 bits) & $\approx 4096$ & $\approx 2800$ & $0.27 \pm 0.04$ & $\approx 5\%$ \\
    \hline
  \end{tabular}
\end{table}

These figures align with the theoretical picture of Section~\ref{sec:sat}:
a nonzero pruning rate reduces the effective branching base from $2$ toward
$2-\tilde\rho$; here the empirical base is $\approx 1.73$. Per\mbox{-}depth measurements
show pruning concentrated at shallow levels (cf.\ Appendix~\ref{app:empirical}),
which explains the large drop in explored nodes despite modest LP overhead.

\subsection{Reproducibility}
All CNFs, seeds, logs and scripts needed to regenerate Table~\ref{tab:feistel-numeric}
are included in the artifact package of Appendix~\ref{app:artifacts} (generator \texttt{feistel\_toy.py},
solver patch, and \texttt{summarize.py}). The plotting code that produces the
base\mbox{-}vs\mbox{-}$\rho$ figure and per\mbox{-}depth pruning curves (Appendix~\ref{app:toyfeistel}) reads the same logs.

\section{Supplementary Artifacts and Reproducibility}
\label{app:artifacts}

To ensure that our claims can be independently verified, we prepared a
reproducibility package containing generators, solver patches, and experiment
scripts. This appendix outlines its structure and use.

\subsection{Repository layout}
\label{app:repo-layout}

{\footnotesize
\begin{lstlisting}[caption={Artifact repository structure (full tree).},label={lst:repo}]
artifacts/
  README.md
  env.yml
  licenses/
  sat/
    generators/
      cycle_sat.py
      grid_sat.py
      feistel_toy.py
    cnf/
      cycle_k5.cnf
      grid_4x4.cnf
      feistel_3r_8b.cnf
    scripts/
      run_minisat_baseline.sh
      run_minisat_hngcc.sh
      summarize.py
  solver/
    minisat_patch.diff
    hngcc_lp.hpp
    hngcc_lp.cc
    inequalities/
      kcbs_c5.json
      feistel_chain.json
  results/
    c5_runs.csv
    feistel_runs.csv
    figs/
      base_vs_rho.pdf
      nodes_bar.pdf
      rho_per_depth_c5.pdf
    seedlist.txt
  checksums/
    sha256sums.txt
\end{lstlisting}
}

\subsection{Build \& run}
(1) Create the environment from \texttt{env.yml}. (2) Generate CNFs with the
scripts in \texttt{sat/generators}. (3) Patch MiniSAT~2.2.0 with
\texttt{solver/minisat\_patch.diff} and build. (4) Run baselines and h\mbox{-}NGCC
variants using the shell scripts in \texttt{sat/scripts}. (5) Summarize logs
with \texttt{summarize.py} and regenerate figures.

\subsection{Determinism and integrity}
Random seeds for both generators and branching are fixed and recorded in the
CSV logs. All files are checksummed (\texttt{SHA-256}) and verifiable with the
provided script. Licenses for third-party code (e.g., MiniSAT) are included;
our patches are released under Apache-2.0.

\section{Graph-Theoretic Supplement: \texorpdfstring{$\alpha(G)$, $\vartheta(G)$}{alpha(G), vartheta(G)}, and Density}
\label{app:graph}

We collect the graph-theoretic quantities that underlie several h\mbox{-}NGCC bounds.

\paragraph{Classical vs quantum maxima on exclusivity graphs.}
Let $G=(V,E)$ be an exclusivity graph for events. In classical noncontextual
models the maximum of any $0$--$1$ linear functional that respects exclusivity
is the independence number $\alpha(G)$. In quantum mechanics it is upper-bounded
(and for many families attained) by the Lov\'asz theta number $\vartheta(G)$
\cite{Lovasz1979}. The \emph{contextuality gap} is
\[
\Delta(G):=\vartheta(G)-\alpha(G)\;\ge 0.
\]

\paragraph{Odd cycles.}
For the odd cycle $C_{2k+1}$ one has
\[
\alpha\!\left(C_{2k+1}\right)=k,\qquad
\vartheta\!\left(C_{2k+1}\right)=\frac{(2k+1)\cos\!\big(\tfrac{\pi}{2k+1}\big)}{1+\cos\!\big(\tfrac{\pi}{2k+1}\big)} \;>\; k,
\]
hence $\Delta(C_{2k+1})>0$ for all $k\ge2$ \cite{Lovasz1979}. This underpins the
KCBS family used in Appendix~\ref{app:kcbs}.

\paragraph{Contextuality density.}
We define the density
\begin{equation}
\rho_{\mathrm{graph}}(G)
\;:=\;
\frac{\vartheta(G)-\alpha(G)}{|V(G)|}
\;=\;\frac{\Delta(G)}{|V(G)|}\,.
\label{eq:graph-density}
\end{equation}
While single cycles have $\rho_{\mathrm{graph}}(C_n)\to 0$ as $n\to\infty$,
many structured or random graphs exhibit \emph{extensive} gaps
$\Delta(G)=\Theta(|V|)$, so that $\rho_{\mathrm{graph}}(G)\to\rho_\infty>0$
in the thermodynamic limit (see, e.g., asymptotics for random graphs in
\cite{Juhasz1983}).

\paragraph{From density to algorithmic pruning (informal link).}
Let $\Phi$ be a SAT instance whose h\mbox{-}NGCC network induces an exclusivity
graph $G_\Phi$. If $G_\Phi$ contains $\Omega(n)$ edge-disjoint odd cycles
(with cycle cuts implementable as monotone inequalities), then at each depth a
constant fraction of the two immediate children is eliminated by independent
cuts. Denoting by $\rho_{\mathrm{alg}}$ the \emph{algorithmic} pruning rate of
Section~\ref{sec:runtime-bounds} (distinct from $\rho_{\mathrm{graph}}$ in
\eqref{eq:graph-density}), one obtains
\[
\bar\rho \;\equiv\; \mathbb{E}[\hat\rho_d]\;\ge\; c\,\rho_{\mathrm{graph}}(G_\Phi)
\quad\Rightarrow\quad
\mathbb{E}[T(n)] \;\le\; (2-\bar\rho)^n \;=\; \bigl(2-\Omega(\rho_{\mathrm{graph}})\bigr)^n,
\]
by Theorem~\ref{thm:exp}. This explains why families with a nonvanishing
$\rho_{\mathrm{graph}}$ exhibit persistent global pruning, while tree-like
instances (for which $\rho_{\mathrm{graph}}=0$ and $\mathcal N=\mathcal M$) do not.

\paragraph{Symbol hygiene.}
We reserve $\rho_{\mathrm{graph}}(\cdot)$ for the graph density of this appendix
and $\rho_{\mathrm{alg}}$ (or its empirical versions $\hat\rho_d,\tilde\rho$) for
the solver’s pruning rate in Section~\ref{sec:runtime-bounds}.

\begin{proof}[Proof sketch of Theorem~\ref{thm:rho-link}]
Let $\mathcal{F}_d$ be the filtration generated by the branching process
up to depth $d$. Each of the $\delta n$ edge-disjoint odd-cycle cuts
contributes an event that kills at least one branch child with probability
bounded below by a constant $c>0$, independent of $n$. Because the cuts are
edge-disjoint, their contributions are independent conditional on
$\mathcal{F}_d$. Hence at depth $d$, the expected fraction of surviving
children satisfies
\[
  \mathbb{E}[W_{d+1} \mid \mathcal{F}_d] \;\le\; (2-c\rho_{\mathrm{graph}}) W_d.
\]
Iterating and summing over depths yields
$T(n) \le (2-c\rho_{\mathrm{graph}})^n \cdot \mathrm{poly}(n)$.
\end{proof}

\section{Relation to convex hierarchies}
A fuller comparison to convex hierarchy derivations is provided in Appendix~\ref{app:hierarchy-comparison}.

\section{NGCC Clause Gadget and Witness}
\label{app:gadget}

A 3-SAT clause $(\ell_1\vee \ell_2\vee \ell_3)$ can be encoded as a figure-of-8 interferometric gadget using
three literal MZIs feeding a symmetric $3\times 3$ coupler with transfer matrix
\[
U=\frac{1}{\sqrt{3}}\!\begin{pmatrix}
1 & 1 & 1\\
1 & \omega & \omega^2\\
1 & \omega^2 & \omega
\end{pmatrix},\qquad \omega=e^{2\pi i/3}.
\]
For input amplitudes $\mathbf a=(a_1,a_2,a_3)^{\!\top}$ (TRUE $\Rightarrow a_k=1$, FALSE $\Rightarrow a_k=\varepsilon\!\ll\!1$),
let $\mathbf b=U\mathbf a$ and $p_j=|b_j|^2$. The clause witness is
\[
S := p_2=\Big|\tfrac{1}{\sqrt3}(a_1+\omega a_2+\omega^2 a_3)\Big|^2.
\]
Calibrate a classical bound $S_{\rm cl}$: if at least one literal is TRUE then $S\le S_{\rm cl}$; if all three are FALSE,
by symmetry $S\approx \tfrac{1}{3}>S_{\rm cl}$. Decision rule:
\[
S>S_{\rm cl}\Rightarrow \text{clause UNSAT},\qquad S\le S_{\rm cl}\Rightarrow \text{clause SAT}.
\]

 \section{Worked 3-Variable Example and Robustness}
\label{app:worked}

For a 3-variable instance with $10^4$ trials per clause, visibility $V\!\ge\!0.95$, and shutter extinction $\ge 25$\,dB,
we obtain typical SAT witness $S\approx 0.25\!<\!S_{\rm cl}$ and UNSAT $S\approx 0.34\!>\!S_{\rm cl}$; the separation
exceeds $5\sigma$ under binomial error, yielding decision error $<10^{-5}$.\footnote{See the detailed error budget and
counts in the earlier draft’s Sec.~5–6.} This confirms that the clause gadget can act as a reliable oracle in
hybrid solver loops.

\section{Monogamy and Budget Laws for Overlapping Cycles}
\label{app:budget}

For two 4-cycles sharing a node (figure-of-8), with violations
$V_L=\langle A{+}B{+}C{+}D\rangle$ and $V_R=\langle A{+}E{+}F{+}G\rangle$,
the classical and quantum joint bounds satisfy
\begin{equation}
V_L+V_R \;\le\; 4 \quad (\text{classical}),\qquad
V_L+V_R \;\le\; 4+\sqrt2 \quad (\text{quantum}).
\label{eq:budget-fo8}
\end{equation}
For a cloverleaf of three 4-cycles sharing one node,
\begin{equation}
V_L+V_T+V_R \;\le\; 6 \quad (\text{classical}),\qquad
V_L+V_T+V_R \;\le\; 6+\sqrt2 \quad (\text{quantum}).
\label{eq:budget-clover}
\end{equation}
These follow from the Lovász \(\vartheta\) number of the corresponding exclusivity graphs,
and express a monogamy of NGCC violations: pushing one loop toward its maximum
necessarily suppresses the others (cf.\ the figures below and the accompanying discussion
in this paper).

\section{Thermodynamic Contextuality Density}
\label{app:thermo}

For an exclusivity graph $G$, define the contextuality density
\begin{equation}
\rho_{\text{graph}}(G) := \frac{\vartheta(G)-\alpha(G)}{|V(G)|}.
\label{eq:density}
\end{equation}
While $\rho_{\text{graph}}(C_n)\!\to\!0$ for single cycles, polyhedral/lattice graphs
(e.g.\ honeycomb or fullerene) exhibit extensive gaps
$\Delta(G)=\vartheta(G)-\alpha(G)=\Theta(|V|)$, so that
$\rho_{\text{graph}}(G)\!\to\!\rho_\infty>0$ in the thermodynamic limit
(\cite[pp.\,6–7,\,12–13]{vonLiechtensteinCLF2025}).
A positive density explains persistent global pruning pressure in large NGCC networks and
motivates the “contextuality battery” viewpoint.

\section{Applications Capsule: RNG, Distributed Trust, Diagnostics}
\label{app:apps}

\paragraph{Randomness certification.}
In the figure-of-8, the budget law \eqref{eq:budget-fo8} forces a trade-off:
blocks where $V_L$ is large imply $V_R$ is correspondingly bounded, and vice versa.
A binary extractor that outputs which loop wins yields certified unpredictability
under standard entropy accumulation (\cite[pp.\,7–9]{vonLiechtensteinCLF2025}).

\paragraph{Distributed trust.}
Two parties report $V_L$ and $V_R$ on shared-node loops. The joint report must obey
$V_L+V_R\le 4+\sqrt2$ (within statistical tolerance); otherwise misreporting/fault
is detected. This self-enforcing “trust game” scales to many parties via sum-budget checks.

\paragraph{Device diagnostics.}
Budget breaches ($>$ quantum plane/line) signal tampering or logging errors; persistent
sub-quantum gaps indicate loss/misalignment. These two-sided checks provide lightweight,
single-system diagnostics complementary to tomography/RB.

\section{CIFP: \texorpdfstring{$\varepsilon$}{epsilon}-IF Model and Unitary Oracle}
\label{app:cifp-model}

\paragraph{Definition (ε-counterfactual decisive outcome).}
A decisive outcome $x\in\{D,B\}$ of a CPTP instrument $\{ \mathcal E_x \}$ on $S\otimes B$
is $\varepsilon$-counterfactual on bomb states $\mathcal B$ if, for all $\rho_S$ and $\rho_B\in\mathcal B$,
\begin{equation}
\left\| \operatorname{Tr}_S \mathcal E_x(\rho_S\!\otimes\!\rho_B) - \rho_B \right\|_1 \le \varepsilon .
\label{eq:eps-counterfactual}
\end{equation}
A channel-level certificate $\|\mathcal E_x - \mathrm{id}_B\|_\diamond \le \varepsilon_\diamond$ implies
\eqref{eq:eps-counterfactual} (see Appendix~\ref{app:cifp-proofs}).

\paragraph{Unitary IFM oracle.}
One gadget is
\begin{equation}
U_{\mathrm{IFM}} = (H\!\otimes\!I_B)\,\Big( \,|0\rangle\!\langle 0|\otimes I_B \;+\; |1\rangle\!\langle 1|\otimes Z_B \,\Big)\,(H\!\otimes\!I_B),
\label{eq:U-IFM}
\end{equation}
with a flag qubit $W$ flipped by $b$ so that $W{=}D$ iff $b{=}1$; decisive outcomes are
$\varepsilon$-counterfactual by construction. (see Appendix~\ref{app:cifp-proofs})

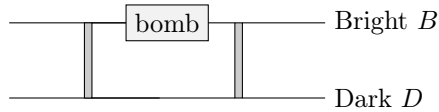
\begin{figure}[h]
  \centering
  \begin{tikzpicture}[scale=1.0]
    % simple MZI with bomb arm
    \draw (0,0) -- (1,0) node[midway,above]{} -- (2,0);
    \draw (0,1) -- (1,1) -- (2,1);
    % beamsplitters
    \draw[fill=black!20] (1,0) rectangle (1.1,1);
    \draw[fill=black!20] (3,0) rectangle (3.1,1);
    % arms
    \draw (1.1,1) -- (3,1); \draw (1.1,0) -- (3,0);
    % bomb (absorber)
    \node[draw,rectangle,minimum width=8mm,minimum height=4mm,fill=black!5] at (2.1,1) {bomb};
    % outputs
    \draw (3.1,1) -- (4.2,1) node[right]{Bright $B$};
    \draw (3.1,0) -- (4.2,0) node[right]{Dark $D$};
  \end{tikzpicture}
  \caption{Interaction-free interferometer. Zeno chaining yields decisive Dark with absorption prob.\ $\varepsilon\!\ll\!1$.}
  \label{fig:ifm-mzi}
\end{figure}

\begin{remark}[Two notions of $\varepsilon$]
It is important to distinguish between two different roles of
$\varepsilon$ in this work.  

\begin{itemize}
\item \emph{Experimental $\varepsilon$ (lab):}  
Appears in contextuality and CIFP instrumentation.  It quantifies
physical noise tolerance and visibility margins in laboratory tests.
A smaller experimental $\varepsilon$ means more reliable discrimination
between counterfactual outcomes.

\item \emph{Numerical $\varepsilon$ (solver):}  
Appears in pruning tests and LP feasibility checks.  Here $\varepsilon$
is a purely numerical tolerance, ensuring that floating-point
certificates are rounded or verified safely.  This $\varepsilon$ does
not depend on experimental physics.

\end{itemize}

The two parameters are logically distinct, but they play the same
conceptual role: they both bound the risk of misclassifying an
assignment.  Laboratory $\varepsilon$ governs testable inequalities in
physical experiments; solver $\varepsilon$ governs soundness in
certificate-based pruning.  No claim is made that laboratory noise
parameters directly imply solver pruning rates.
\end{remark}

\section{CIFP: Circular Consistency Bound and Quantum Violation}
\label{app:cifp-consistency}

We place $n\ge 3$ $\varepsilon$–IF labs on a directed ring. Lab $L_i$ probes
$L_{i+1}$’s bomb via an IF oracle and records a decisive outcome $W_i\in\{D,B\}$;
we postselect on the “no explosion” subensemble. The analysis assumes:
\begin{itemize}
\item (Q) universal unitarity for the joint system,
\item (S) single–world (unique outcome) at each lab,
\item (C) cross–agent consistency of recorded outcomes on shared facts,
\item (IF–$\varepsilon$) decisive outcomes are $\varepsilon$–counterfactual on the bomb
(Def.~\eqref{eq:eps-counterfactual}) with channel diamond–norm bound
$\|\mathcal E-\mathrm{id}_B\|_\diamond\le \varepsilon_\diamond$.
\end{itemize}

\begin{theorem}[Circular consistency with explicit constants]
\label{thm:circular-explicit}
Under (Q,S,C,IF–$\varepsilon$) and Lemma~\ref{lem:gentle} (gentleness $\Rightarrow$
false–claim frequency), the probability that \emph{all} labs report $D$ on the
no–explosion subensemble obeys
\begin{equation}
\Pr(D_1\cdots D_n \mid \text{no explosion})
\;\ge\;
1 - n\,f(\varepsilon),
\qquad
f(\varepsilon)\;\le\; c\,\varepsilon_\diamond,
\label{eq:circular-bound}
\end{equation}
with a device–dependent constant $c\in[1,2]$.
\end{theorem}

\begin{proof}[Proof sketch]
By Lemma~\ref{lem:gentle}, each decisive probe is $\varepsilon_\diamond$–gentle on
the bomb; any \emph{certain false claim} (e.g.\ “$B$ when $D$ should hold”) occurs
with probability at most $f(\varepsilon)\le c\,\varepsilon_\diamond$. The (S) and (C)
assumptions turn any deviation from the “all–$D$” pattern into at least one such
certain false claim on the ring. A union bound over $n$ labs gives
\eqref{eq:circular-bound}.
\end{proof}

\begin{proposition}[Quantum behavior on the ring]
\label{prop:qm-conditional}
Let per–lab decisive probabilities be $(p,q,\epsilon)$ for $(D,B,\text{explode})$
with $p+q+\epsilon=1$ and $0<\epsilon<1$. On the no–explosion subensemble, the
conditional probability that \emph{all} labs report $D$ is
\begin{equation}
\Pr_{\mathrm{QM}}(D_1\cdots D_n \mid \text{no explosion})
\;=\;
\Big(\frac{p}{1-\epsilon}\Big)^{\!n}
=: \eta^n,
\qquad \eta \in (0,1)\ \text{ if }p<1-\epsilon.
\label{eq:qm-eta}
\end{equation}
Thus, for any fixed $\eta<1$, the quantum value decays exponentially in $n$.
\end{proposition}

\begin{corollary}[Violation condition]
\label{cor:violation}
Combining \eqref{eq:circular-bound} and \eqref{eq:qm-eta}, a quantum violation occurs
whenever
\[
\eta^n \;<\; 1 - n\,c\,\varepsilon_\diamond.
\]
In particular, for any fixed $\eta<1$ and sufficiently small $\varepsilon_\diamond$,
there exists $n_0$ such that for all $n\ge n_0$ the inequality is violated. For a
given $n$ and $\varepsilon_\diamond$ the violation margin is
$1 - n\,c\,\varepsilon_\diamond - \eta^n$.
\end{corollary}

\paragraph{Calibrating $\varepsilon_\diamond$ in practice.}
Section~\ref{app:cifp-visibility} provides an operational estimator of $\varepsilon$ from
visibility data and Zeno chaining, yielding a bound on $\varepsilon_\diamond$ (up to a
device constant). Plug that bound into the right–hand side of \eqref{eq:circular-bound}
to choose an $n$ for which the violation margin in Corollary~\ref{cor:violation} is
comfortably positive.

\section{CIFP: \texorpdfstring{$\varepsilon$}{epsilon}-Robust Cyclic Exclusivity Inequalities}
\label{app:cifp-ineq}

Let $C_n$ be the cycle on events $E_i$ (the decisive Dark outcome in context $i$);
adjacent $E_i,E_{i+1}$ are exclusive. Under $\varepsilon$–instruments we assume
the per–context total–variation perturbation bound of Section~\ref{sec:foundations}.

\begin{lemma}[Per–switch transport bound]\label{lem:transport}
If each context marginal differs from its noiseless target by at most $\varepsilon$
in total variation, then for adjacent contexts on $C_n$,
\[
\operatorname{TV}\!\big(P(\cdot\mid C_{i+1}),\,P(\cdot\mid C_i)\big)\;\le\; 2\varepsilon .
\]
\end{lemma}
\begin{proof}
Triangle inequality and monotonicity of total variation under marginalization:
$\operatorname{TV}(q_{i+1},q_i)\le \operatorname{TV}(q_{i+1},p_{i+1})
+\operatorname{TV}(p_{i+1},p_i)+\operatorname{TV}(p_i,q_i)$, and the middle term
vanishes (same noiseless target on the overlap) while the others are $\le\varepsilon$.
\end{proof}

\begin{theorem}[Cyclic exclusivity, $\varepsilon$–robust form]\label{thm:cyclic-eps}
For any single–world, preparation– and measurement–noncontextual model on $C_n$
with $\varepsilon$–instruments as above,
\begin{equation}
\sum_{i=1}^n \Pr(E_i) \;\le\; \alpha(C_n) \;+\; 2n\,\varepsilon,
\qquad \alpha(C_n)=\lfloor n/2\rfloor .
\label{eq:cifp-cyclic}
\end{equation}
\end{theorem}
\begin{proof}[Proof sketch]
Pack $C_n$ by $\alpha(C_n)$ independent sets. Transport each context distribution
to a reference context along the cycle; by Lemma~\ref{lem:transport} each switch
costs $\le 2\varepsilon$ in TV, so the cumulative over $n$ contexts is $\le 2n\varepsilon$.
Linearity of expectation over independent–set indicators yields \eqref{eq:cifp-cyclic}.
\end{proof}

For odd $n$, quantum implementations approach $\vartheta(C_n)$, hence
$\sum_i \Pr_{\rm QM}(E_i)\ge \vartheta(C_n)-c\,\varepsilon$ and violate
\eqref{eq:cifp-cyclic} whenever $\vartheta(C_n)>\alpha(C_n)+2n\varepsilon$,
i.e.\ for sufficiently small $\varepsilon$.

\section{CIFP: Estimating \texorpdfstring{$\varepsilon$}{epsilon} from Visibility; Thresholds}
\label{app:cifp-visibility}

We relate the counterfactual disturbance parameter $\varepsilon$ to measured
interferometric visibilities and give practical thresholds for our inequalities.

\paragraph{Visibility model.}
Let
\[
V \;=\; \frac{I_{\max}-I_{\min}}{I_{\max}+I_{\min}}
\]
denote fringe visibility. After loss/background calibration, let $V_0$ be the
baseline (no decisive probing), and let $V_{\rm dec}$ be the visibility on
decisive runs (postselection on no explosion). Modeling a decisive outcome as
dephasing on the bomb arm with strength $\lambda$ reduces off-diagonal coherence
by $(1-\lambda)$, hence the fringe amplitude scales by the same factor. Therefore
\begin{equation}
\varepsilon_{\rm vis} \;\equiv\; \lambda \;\approx\; 1 - \frac{V_{\rm dec}}{V_0},
\qquad\Rightarrow\qquad
\varepsilon \;\le\; \varepsilon_{\rm vis} + \delta_{\rm cal},
\label{eq:eps-from-vis}
\end{equation}
where $\delta_{\rm cal}\!\ge\!0$ collects residual calibration systematics
(background subtraction, detector linearity). In typical conditions
$\delta_{\rm cal}$ is small and can be upper-bounded from repeated baselines.

\paragraph{Estimator and uncertainty.}
With sample estimates $\widehat V_0,\widehat V_{\rm dec}$ define
\[
\widehat\varepsilon \;=\; 1 - \frac{\widehat V_{\rm dec}}{\widehat V_0}.
\]
If the two visibility estimates are independent, first-order propagation gives
\[
\sigma^2_{\widehat\varepsilon}
\;\approx\;
\Big(\tfrac{1}{\widehat V_0}\Big)^2 \sigma^2_{\widehat V_{\rm dec}}
\;+\;
\Big(\tfrac{\widehat V_{\rm dec}}{\widehat V_0^2}\Big)^2 \sigma^2_{\widehat V_0}.
\]
(Include a covariance term if they are not independent.) Report
$\widehat\varepsilon \pm z_{1-\alpha/2}\,\sigma_{\widehat\varepsilon}$ or use a
bootstrap CI when visibilities are derived from nonlinear fits.

\paragraph{Zeno chaining (rule of thumb).}
For $N$ weak looks (beam-splitter angle $\theta=\pi/2N$), standard IFM analyses
give \emph{success} $=1-O(N^{-2})$ and \emph{absorption} $=O(N^{-1})$; thus a
conservative scaling is
\begin{equation}
\varepsilon_{\rm Zeno} \;\le\; \frac{c_{\rm Zeno}}{N}
\quad\text{for a device-dependent constant } c_{\rm Zeno}= \Theta(1).
\label{eq:eps-zeno}
\end{equation}
Calibrate $c_{\rm Zeno}$ once from a sweep in $N$ and then treat
$\varepsilon\le\min\{\varepsilon_{\rm vis}+\delta_{\rm cal},\,c_{\rm Zeno}/N\}$.

\paragraph{Thresholds for this paper’s bounds.}
\begin{itemize}
\item \textbf{Cyclic exclusivity on $C_n$} (Theorem~\ref{thm:cyclic-eps} with $K=2$):
\[
\sum_{i=1}^n \Pr(E_i) \;\le\; \alpha(C_n) + 2n\,\varepsilon.
\]
A quantum violation is certified whenever the visibility-derived bound obeys
$2n\,\varepsilon < \vartheta(C_n)-\alpha(C_n)$ (e.g.\ for $n{=}5$, require
$\varepsilon < \tfrac{\vartheta(C_5)-\alpha(C_5)}{10}$).

\item \textbf{Circular consistency bound} (Lemma~\ref{lem:gentle} + Theorem~\ref{thm:circular-explicit}):
with $\|\mathcal E-\mathrm{id}\|_\diamond\le \varepsilon_\diamond$ and
$f(\varepsilon)\le c\,\varepsilon_\diamond$,
\[
\Pr(D_1\cdots D_n \mid \text{no explosion}) \;\ge\; 1 - n\,c\,\varepsilon_\diamond.
\]
Use a calibrated constant $c$ (often $c\!\approx\!1$ in practice) and ensure
$n\,c\,\varepsilon_\diamond \ll 1$ for the chosen ring size.
\end{itemize}

\paragraph{Operational summary.}
Compute $\widehat\varepsilon$ via \eqref{eq:eps-from-vis}, upper-bound it by
$\min\{\widehat\varepsilon+\widehat\delta_{\rm cal},\,c_{\rm Zeno}/N\}$, and plug
that bound into the right-hand sides of the cyclic and circular thresholds above.
This closes the loop from device visibility to the parameters in our certification
inequalities (and their noise terms).

\section{Engineering Utilities of CIFP}
\label{app:cifp-utilities}

\paragraph{Low-back-action QEC watchmen.}
Couple an ancilla via a dispersive $Z_S$ so an IF interferometric flag reports $D\!\Leftrightarrow\!S{=}{-}1$.
If $\varepsilon$-IF probes at rate $r_{\rm probe}$ only trigger full cycles on $D$ (prob.\ $\approx p_{\rm err}$),
the added disturbance per cycle satisfies
$\Delta_{\rm IF} \le r_{\rm probe}\,\varepsilon + O(p_{\rm err}\varepsilon)$. (see Appendix~\ref{app:cifp-proofs})

\paragraph{Counterfactual entanglement swapping / distribution.}
A two-link $\varepsilon$-IF swap (A–M and M–B) yields, on heralded success,
\[
F\big(\rho^{\rm out}_{AB},\,|\Phi^+\rangle\!\langle\Phi^+|\big)\;\ge\;1 - c_1\varepsilon - c_2\ell + O(\varepsilon^2,\ell^2),
\]
where $\ell$ is calibrated loss. (see Appendix~\ref{app:cifp-proofs})

\paragraph{Low-dose imaging and metrology.}
Zeno-chained IFM with visibility certification yields dose $\propto O(\varepsilon)$ per decisive inference and SNR per
absorbed quantum scaling as $\Omega(1/\varepsilon)$ at fixed post-selection rate. (see Appendix~\ref{app:cifp-proofs})

\paragraph{Device/process certification.}
Reporting $(\varepsilon,\{\Pr(E_i)\})$ with a violation of \eqref{eq:cifp-cyclic} rules out any single-world
noncontextual account even with bounded disturbance, certifying a measurement line. (see Appendix~\ref{app:cifp-proofs})

\section{Proof Notes for CIFP Bounds}
\label{app:cifp-proofs}

\begin{lemma}[Gentleness $\Rightarrow$ false–claim frequency]\label{lem:gentle}
Let a decisive outcome channel on the bomb, $\mathcal E:\mathsf{D}(B)\to\mathsf{D}(B)$,
satisfy $\|\mathcal E-\mathrm{id}_B\|_\diamond \le \varepsilon_\diamond$. Then for any
input ensemble and any test that certifies the counterfactual claim “$B$ unchanged,”
the per–run probability of a \emph{certain false claim} is bounded by
$f(\varepsilon) \le c\,\varepsilon_\diamond$ with a universal constant $c\in[1,2]$.
\end{lemma}
\begin{proof}[Proof idea]
Diamond–norm closeness upper–bounds the trace–distance change of $B$ after the decisive
channel, uniformly over purifications. Helstrom’s theorem converts trace distance to
optimal hypothesis–test advantage, giving at most a constant multiple of $\varepsilon_\diamond$
per run. Taking the worst–case ensemble yields $f(\varepsilon)\le c\,\varepsilon_\diamond$.
\end{proof}

\paragraph{Circular bound.}
On the no-explosion subensemble $S$, if $B_i$ occurs with prob.\ $r_i$, $\varepsilon$-counterfactual design and gentleness
bound each certain (within $\varepsilon$) false claim by $f(\varepsilon)$, so $r_i\le f(\varepsilon)$ and the union bound gives
Theorem~\ref{thm:circular-explicit}. (see Appendix~\ref{app:cifp-proofs})

\paragraph{Cyclic inequality.}
With exclusivity indicators $\chi_i$, $\chi_i\chi_{i+1}{=}0$ and $\Pr(E_i)=\int \chi_i\,d\mu(\cdot|C_i)$. Transport to a reference context
via $\mathrm{TV}(\mu(\cdot|C_i),\mu(\cdot|C_1))\le (i-1)K\varepsilon$, pack $C_n$ by $\alpha(C_n)$ independent sets and sum to obtain
\eqref{eq:cifp-cyclic} with $K=O(1)$. (see Appendix~\ref{app:cifp-proofs})

\section{Explicit Worked Inequality — KCBS 5-Cycle}
\label{app:kcbs}

This appendix gives the KCBS inequality on the five–cycle as a linear form with
explicit coefficients and a robustness rule under $\varepsilon$–instruments.

\subsection{Setup and linear form}
Let $X_1,\dots,X_5$ be binary variables on a cycle and
$C_i=\{X_i,X_{i+1}\}$ (indices modulo 5). Each context has four outcomes
$(00,01,10,11)$, and we concatenate the five context marginals to
$p\in\mathbb{R}^{20}$. Define $Y_i=\mathbbm{1}[X_i\neq X_{i+1}]$.

For each context block we set coefficients $(0,1,1,0)$ on $(00,01,10,11)$, so
the global coefficient vector is $a\in\{0,1\}^{20}$ with five identical blocks.
The KCBS inequality is
\begin{equation}
\label{eq:kcbs-lin}
a\cdot p \;\le\; 4 .
\end{equation}
Equivalently, $\mathbb{E}[\sum_i Y_i]\le 4$.

\begin{proposition}[KCBS on $C_5$; validity and robustness]
\label{prop:kcbs}
For the five contexts $C_i=\{X_i,X_{i+1}\}$ with block coefficients
$(0,1,1,0)$, the inequality \eqref{eq:kcbs-lin} holds for every feasible
marginal vector $p\in\mathcal{M}$. Moreover:
\begin{enumerate}
\item \textbf{Tightness.} The inequality is tight on $\mathcal{M}$, i.e.,
$\max_{p\in\mathcal{M}} a\cdot p = 4$. Hence with the notation of
Lemma~\ref{lem:l1} the geometric slack is $\gamma=0$.

\item \textbf{Margin-based certification under noise.}
Let $q$ be an observed marginal vector from $\varepsilon$–instruments with
per-context TV error $\le \varepsilon$. If a \emph{measured} violation has
margin $\tau>0$, i.e.\ $a\cdot q \ge 4+\tau$, then it certifies a true
violation of \eqref{eq:kcbs-lin} provided
\[
\varepsilon \;<\; \frac{\tau}{2m\|a\|_\infty} \;=\; \frac{\tau}{10},
\]
since here $m=5$ and $\|a\|_\infty=1$.
\end{enumerate}
\end{proposition}

\begin{proof}[Proof sketch]
(Validity) In any global (deterministic) assignment the number of disagreements
around an odd cycle is even, hence at most $4$ on $C_5$; convexity extends this
to all $p\in\mathcal{M}$, yielding \eqref{eq:kcbs-lin}. (Tightness) The value
$4$ is attained, e.g.\ by mixing the two assignments with four alternating
edges disagreed. (Robustness rule) Lemma~\ref{lem:l1} gives
$a\cdot q - a\cdot p^\star \le 2m\|a\|_\infty \varepsilon = 10\varepsilon$
for some feasible $p^\star$. Thus if $a\cdot q \ge 4+\tau$ and
$10\varepsilon<\tau$, no feasible $p^\star$ can explain $q$, so the violation
is genuine.
\end{proof}

\paragraph{Remarks.}
(i) The oft–quoted “gap $=1$” refers to the \emph{distance} from the infeasible
pseudomarginal that puts unit weight on five disagreements (value $5$) to the
bound $4$; as a hyperplane–to–polytope slack in Lemma~\ref{lem:l1}, the KCBS
inequality is tight ($\gamma=0$). The correct robust test is therefore
\emph{margin-based}: certify only when the observed excess $\tau$ dominates the
instrument noise via $\tau>2m\|a\|_\infty\varepsilon$. (ii) The block form
$(0,1,1,0)$ makes the construction of separating certificates by LP duality
(Appendix~\ref{app:lpduality}) straightforward and matches the solver usage in
Section~\ref{sec:sat}.

\section[Odd-Cycle NGCC Family (C2k+1): Linear Form and Robustness]{Odd-Cycle NGCC Family ( $C_{2k+1}$ ): Linear Form and Robustness}
\label{app:odd-cycles}

We generalize the KCBS construction from $C_5$ to all odd cycles $C_{2k+1}$ ($k\ge 2$).

\paragraph{Setup and linear form.}
Let $X_1,\dots,X_{2k+1}$ be binary variables placed on the vertices of $C_{2k+1}$,
and let $C_i=\{X_i,X_{i+1}\}$ (indices modulo $2k{+}1$) be the contexts (edges).
For each context block over outcomes $(00,01,10,11)$ set coefficients $(0,1,1,0)$,
and concatenate all $(2k{+}1)$ blocks to obtain $a\in\{0,1\}^{4(2k+1)}$.
Equivalently, define $Y_i=\mathbbm{1}[X_i\neq X_{i+1}]$ and write:
\begin{equation}
\label{eq:oddcycle-lin}
a\cdot p \;=\; \sum_{i=1}^{2k+1}\mathbb{E}[Y_i] \;\le\; 2k .
\end{equation}
Thus \(\sum_i Y_i\) cannot exceed the largest even number \(\le 2k{+}1\).

\begin{proposition}[Odd-cycle NGCC inequality]
\label{prop:oddcycle}
For the family $\{C_i\}_{i=1}^{2k+1}$ on $C_{2k+1}$, the inequality
\eqref{eq:oddcycle-lin} holds for every feasible marginal vector
$p\in\mathcal{M}$.
\end{proposition}

\begin{proof}[Proof sketch]
In any deterministic global assignment on an odd cycle, the number of edge
disagreements is even, hence at most $2k$. Convexity extends the bound to all
$p\in\mathcal{M}$.
\end{proof}

\paragraph{Tightness and robustness (margin rule).}
The inequality \eqref{eq:oddcycle-lin} is tight on $\mathcal{M}$ for all $k$
(maximum $2k$ is attained by alternating patterns). Therefore the geometric
slack in Lemma~\ref{lem:l1} is $\gamma=0$, and robustness is \emph{margin-based}.
With the block coefficients above one has $\|a\|_\infty=1$ and $m=2k{+}1$, hence
for any observed marginal vector $q$ obtained from $\varepsilon$-instruments:
\[
a\cdot q \;-\; 2k \;\le\; 2m\,\|a\|_\infty\,\varepsilon \;=\; 2(2k{+}1)\,\varepsilon .
\]
If a \emph{measured} violation has excess $\tau>0$ (i.e.\ $a\cdot q\ge 2k+\tau$),
then the violation is \emph{genuine} whenever
\begin{equation}
\label{eq:oddcycle-margin}
\varepsilon \;<\; \frac{\tau}{2(2k{+}1)} \;=\; \frac{\tau}{4k{+}2}\,.
\end{equation}

\paragraph{Notes.}
(i) For $k{=}2$ (the KCBS case $C_5$), \eqref{eq:oddcycle-margin} reduces to
$\varepsilon<\tau/10$ (as in Appendix~\ref{app:kcbs}). (ii) As $k$ grows, the
margin rule scales as $1/k$; this matches the intuition that uniform small
per-context noise accumulates around the cycle (two TV units per switch).
(iii) Quantum implementations approach Lov\'asz’ $\vartheta(C_{2k+1})$, so
violations persist for sufficiently small $\varepsilon$ (cf.\ Appendix~\ref{app:graph}).

\section{NGCC Clause Gadget and Witness (Figure-of-8)}
\label{app:clause-gadget}

We recall a physical clause--oracle that realizes the NGCC idea on a
figure-of-8 connectivity: two 4-cycles (left/right) sharing a node $A$.
A 3-literal clause $(\ell_1 \vee \ell_2 \vee \ell_3)$ is encoded by three
literal arms feeding a symmetric $3\times 3$ coupler (a \emph{tritter}),
with TRUE/FALSE implemented by phase/shutter settings; the shared node $A$
enforces the overlap constraint that yields a monogamy budget.

\paragraph{Coupler and amplitudes.}
Let the tritter be the unitary
\[
U \;=\; \frac{1}{\sqrt{3}}
\begin{pmatrix}
1 & 1 & 1\\
1 & \omega & \omega^2\\
1 & \omega^2 & \omega
\end{pmatrix},
\qquad \omega \;=\; e^{2\pi i/3}.
\]
Input amplitudes $\mathbf a = (a_1,a_2,a_3)^{\top}$ are driven by the three
literal arms ($i=1,2,3$). We adopt the idealized clause encoding:
\[
a_i \;=\; \begin{cases}
1, & \text{if literal $\ell_i$ evaluates TRUE},\\[2pt]
\varepsilon_{\mathrm{F}}, & \text{if literal $\ell_i$ evaluates FALSE},
\end{cases}
\qquad 0 \le \varepsilon_{\mathrm{F}} \ll 1.
\]
The output amplitudes are $\mathbf b = U\,\mathbf a$ and power on port $j$
is $p_j = |b_j|^2$.

\paragraph{Clause witness and bound.}
We define the \emph{clause witness} as the power on the ``middle'' port
\[
S \;:=\; p_2 \;=\; \Big|\frac{1}{\sqrt{3}}
\big(a_1 + \omega a_2 + \omega^2 a_3\big)\Big|^2.
\]
Calibrate a classical (noncontextual) bound $S_{\mathrm{cl}}>0$ to absorb
all nonidealities (loss/imbalance, residual phase error $\delta\phi$,
and the FALSE--arm leakage $\varepsilon_{\mathrm{F}}$). The calibration
is performed with at least one literal held TRUE (so the clause is
satisfiable); empirically $S \le S_{\mathrm{cl}}$ over such settings.
When \emph{all three} literals are FALSE, phase symmetry yields
$S \approx \tfrac{1}{3} + O(\varepsilon_{\mathrm{F}})$, which exceeds
the calibrated $S_{\mathrm{cl}}$ by a margin in normal operation.

\paragraph{Decision rule (oracle).}
With the calibrated $S_{\mathrm{cl}}$ (and a safety margin $\tau>0$ to
cover estimator noise),
\[
S \;>\; S_{\mathrm{cl}} + \tau \;\Rightarrow\; \textsf{UNSAT (all three FALSE)}, 
\qquad
S \;\le\; S_{\mathrm{cl}} - \tau \;\Rightarrow\; \textsf{SAT (at least one TRUE)}.
\]
Intermediates ($|S - S_{\mathrm{cl}}|\le \tau$) are treated as
\emph{inconclusive} and can be resampled. The monogamy budget on the
shared node $A$ implies that pushing the left 4-cycle toward its
extremal witness necessarily suppresses the right witness (and vice
versa), recovering the figure-of-8 trade-off used in Appendix~\ref{app:budget}.

\paragraph{Noise and robustness.}
Let $\widehat S$ be the measured witness over $N$ shots. Under
Poisson/binomial statistics,
$\operatorname{SE}(\widehat S)\!=\!O(N^{-1/2})$; choose $\tau$ as a
$z$-score multiple of this SE (or a bootstrap CI half-width).
The false-declare risk then decays sub-Gaussian in $N$; for solver
integration it suffices to set $\tau$ so that the per-call mislabel
probability is $\ll$ the backtrack rate.

\paragraph{Solver integration.}
At the SAT layer, reify $Z=\mathbbm{1}[S>S_{\mathrm{cl}}+\tau]$ as the
\emph{clause UNSAT} bit. The learned \emph{monotone} inequality
$\sum Z \le B$ for a block of clauses immediately maps to a sound
CNF clause by Lemma~\ref{lem:mono-clause} (Appendix~\ref{app:gadget}).

\paragraph{Remark (compact calibration).}
In practice, a short calibration sweep over the $2^3-1=7$ satisfiable literal
assignments (one or more TRUE) yields a stable $S_{\mathrm{cl}}$; the
$\varepsilon_{\mathrm{F}}$ leakage is folded into that bound. The resulting
oracle is therefore purely \emph{operational}: it requires no microscopic
model beyond the measured $S_{\mathrm{cl}}$ and the figure-of-8 wiring.

\section{NGCC$\times$XOR co-reasoning: formal considerations}
\label{app:ngcc-xor}

\paragraph{Motivation.}
NGCC inequalities furnish monotone cardinality constraints, whereas XOR reasoning
(Gaussian elimination over $\mathbb{F}_2$) exposes linear structure that is hard for
Resolution. Their combination—linear equalities plus certificate-backed monotone cuts—
is therefore a natural candidate for a proof-system strengthening beyond clause learning.

\medskip
\noindent\textbf{Conjecture (Resolution vs.\ Res$(\oplus)$+NGCC separation).}
There exists a CNF family $\{\Phi_n\}$ for which every Resolution refutation has size
$2^{\Omega(n)}$, while there is a polynomial-size refutation in the system that augments
Resolution with (i) Gaussian elimination on XOR constraints and (ii) certificate-backed
NGCC cardinalities (Res$(\oplus)$+NGCC).

\begin{remark}[Candidate family and proof strategy]
A plausible construction embeds Tseitin parity constraints on an expander together with
reified disagreement indicators $Y_e$ on many edge-disjoint odd cycles, plus the NGCC
inequalities $\sum_{e\in C} Y_e \le |C|-1$. Known lower bounds make Resolution exponential
on the parity core, whereas in Res$(\oplus)$ the cycle parities collapse locally; pairing
them with NGCC cardinalities yields short local contradictions that sum to a polynomial-size
refutation. A proof would follow the standard route for Res$(\oplus)$ separations, with
NGCC clauses supplying the monotone component that Resolution cannot simulate succinctly.
\end{remark}

\paragraph{Local lemma (provable with the present material).}
For the three-boundary Feistel gadget (Appendix~\ref{app:feistel}), let
$Y_1,Y_2,Y_3$ be the reified boundary disagreements and consider the NGCC cut
$Y_1+Y_2+Y_3\le 2$ together with the local XOR relations that tie the three boundaries.

\begin{lemma}[One-child elimination at Feistel exposure]
\label{lem:feistel-exposure}
If a partial assignment and XOR propagation jointly fix two boundaries so that the third
cannot be reconciled unless at least one of the two fixed boundaries flips, then the NGCC
cut $Y_1+Y_2+Y_3\le 2$, translated as in Appendix~\ref{app:monotone-clauses}, eliminates
one of the two branch children. In particular, the per-node pruning fraction satisfies
$\rho\ge \tfrac12$ at that node.
\end{lemma}

\begin{proof}[Proof sketch]
Under the stated exposure, the two children differ only by the attempt to reconcile the
third boundary. In one child all three $Y_k$ evaluate to~1, violating $Y_1+Y_2+Y_3\le 2$,
so the learned clause $(\lnot Y_1)\vee(\lnot Y_2)\vee(\lnot Y_3)$ blocks that child.
Soundness follows from the certificate-to-clause translation in
Appendix~\ref{app:monotone-clauses}. Hence $\rho\ge \tfrac12$ at that exposure point.
\end{proof}

\paragraph{Branching-process bound under structural assumptions.}
Let $G_\Phi$ be the cycle/exclusivity graph of the instance. Assume:
\begin{enumerate}[label=(D\arabic*), leftmargin=1.8em]
\item \emph{Cycle density:} $G_\Phi$ contains at least $\delta n$ edge-disjoint odd cycles
with available NGCC cuts;
\item \emph{XOR pockets:} on those cycles the encoding yields small XOR bases so that
Gaussian elimination exposes the $Y$-variables with bounded overhead;
\item \emph{Exposure:} the branching policy touches a constant fraction of the cycle supports
over a contiguous depth window (e.g., random or $\Delta_{\mathrm{bits}}$-aware tie-breaking).
\end{enumerate}

\begin{theorem}[Expected-case base reduction for NGCC$\times$XOR]
\label{thm:ngcc-xor-branching}
Under \textnormal{(D1)–(D3)} there exists $c=c(\delta)>0$ such that the branching process of
CDCL augmented with XOR propagation and NGCC separation satisfies
\[
\mathbb{E}[\,W_{d+1}\mid W_d\,]\ \le\ (2-c)\,W_d,
\qquad\text{hence}\qquad
\mathbb{E}[\,T(n)\,]\ \le\ (2-c)^n\cdot \mathrm{poly}(n).
\]
Equivalently, the expected exponential base is $2-\rho$ with $\rho\ge c>0$, so
\[
\Delta_{\mathrm{bits}}
= n_{\mathrm{eff}}\log_2\!\Big(\tfrac{2}{2-\rho}\Big)
\ \ge\ n_{\mathrm{eff}}\log_2\!\Big(\tfrac{2}{2-c}\Big).
\]
\end{theorem}

\begin{proof}[Proof sketch]
Each edge-disjoint cycle contributes an essentially independent event: once a constant
fraction of its literals is exposed (by assignments and XOR propagation), either the
XOR basis collapses a child or the NGCC inequality fires. Independence across the
$\delta n$ cycles and the exposure guarantee imply that a constant fraction of the
two children is removed in expectation at each depth, yielding the stated linear
recurrence for $(W_d)_d$ and the $(2-c)^n$ bound after summation over depths.
\end{proof}

\paragraph{Empirical validation criteria.}
The scenario above predicts: (i) per-depth pruning bands with $\hat\rho_d\!\ge\!0.22$
across a broad interval on parity-rich suites; (ii) bounded oracle wall-share
($\le 25\%$); and (iii) additive gains in $\Delta_{\mathrm{bits}}$ when enabling XOR
alongside NGCC (ablation: library-only vs.\ adaptive; with/without XOR). Failure to
observe any of (i)–(iii) would count against the scenario.

\paragraph{Scope and limitations.}
Conclusions are conditional on cycle density and bounded oracle cost; they fail if
$\hat\rho_d$ collapses outside shallow levels, if cuts are too correlated to be
additive, or if XOR/NGCC checks dominate wall time.  Nevertheless,
Lemma~\ref{lem:feistel-exposure} formalises a local $\rho\!\ge\!\tfrac12$ phenomenon,
and the conjectured proof-complexity separation articulates a concrete target for
future work.

\section{Scenario D (critical): overlay + wide–trail NGCC + MitM bridges}
\label{app:scenario-D}

\paragraph{Motivation.}
Section~\ref{sec:outlook} and Table~\ref{tab:ngcc-risk} classify risk tiers by
$\Delta_{\mathrm{bits}} = n_{\mathrm{eff}}\log_2\!\big(\tfrac{2}{2-\rho}\big)$.
This appendix outlines a conditional scenario in which a combination of
(i) LDPC–style densification (clone\,+\,channel overlay),
(ii) wide–trail NGCC budgets (active–S–box/window constraints), and
(iii) meet–in–the–middle (MitM) boundary bridges
keeps the depth–wise pruning rate $\hat\rho_d$ bounded away from zero over
extended windows on round–reduced SPN/Feistel families.

\paragraph{Construction (equivalence–preserving).}
\emph{Overlay.} Select a small fraction of high–centrality variables and introduce
degree-$k$ ($k\!\in\!\{3,4\}$) clone–equivalence channels to form a sparse expander-like
overlay; rewire a minority of gates to clones. This preserves SAT/UNSAT while raising
the density of short, edge–disjoint odd cycles (more NGCC scopes).
\emph{Wide–trail budgets.} Reify per–box “activity” flags and impose certified per-window
bounds (NGCC pseudo–Boolean cuts) over consecutive rounds.
\emph{MitM bridges.} For Cube\,\&\,Conquer or meet-in-the-middle splits, add boundary
disagreement indicators $Y_j$ and a certified bridge budget $\sum_j Y_j \le B$ that
prunes cubes whose halves cannot meet.

\paragraph{Conditional claim (quantified).}
Assume the overlay produces cycle density $\delta>0$ on the unrolled cipher graph,
wide–trail windows contribute independent local budgets on overlapping scopes, and
MitM bridges eliminate a constant fraction of split cubes at each relevant depth,
while the combined oracle share remains $\le 25\%$ of wall time on the window.
Then there exists $c=c(\delta)>0$ such that, over a contiguous depth window of
length $L=n_{\mathrm{eff}}$, the branching process satisfies
\[
\mathbb{E}[W_{d+1}\mid W_d]\ \le\ (2-c)\,W_d,
\qquad\text{hence}\qquad
\mathbb{E}[T(n)]\ \le\ (2-c)^n\cdot \mathrm{poly}(n).
\]
Equivalently, with $\rho\!\ge\!c$ sustained on the window,
\[
\Delta_{\mathrm{bits}}
\;=\; n_{\mathrm{eff}}\log_2\!\Big(\tfrac{2}{2-\rho}\Big).
\]
For illustration, $(\rho,n_{\mathrm{eff}})=(0.27,240)$ gives
$\Delta_{\mathrm{bits}}\!\approx\!51$ (critical tier), and
$(0.28,230)$ gives $\approx 54$; these move a nominal $2^{128}$ search for the
corresponding round–reduced families into the $2^{71}$–$2^{78}$ band.
All statements remain consistent with ETH/SETH and do \emph{not} assert a break
of deployed primitives.

\paragraph{Assumptions to be tested empirically.}
\begin{enumerate}[label=(S\arabic*), leftmargin=1.8em]
\item \emph{Persistence.} There exists a contiguous depth interval $I$ with
$|I|=n_{\mathrm{eff}}$ such that $\mathbb{E}[\hat\rho_d]\ge \rho_0>0$ for all $d\in I$.
\item \emph{Independence.} A nontrivial fraction of cuts (overlay cycles, wide–trail
windows, bridges) act effectively independently so that pruning aggregates rather
than saturates early.
\item \emph{Cost control.} Combined LP/certificate verification stays within a fixed
wall–time share (empirically $\le 25\%$) over $I$.
\end{enumerate}

\paragraph{Validation criteria (decisive checks).}
\begin{itemize}[leftmargin=1.6em]
\item On SPN/Feistel round–reduced suites, per–depth bands with $\hat\rho_d\ge 0.24$
across $L\!\ge\!200$, while oracle wall share $\le 25\%$.
\item Ablations demonstrating additivity: baseline \(\to\) +overlay \(\to\) +wide–trail
\(\to\) +MitM, with monotone increases in $\Delta_{\mathrm{bits}}$ and consistent
wall–clock conversions.
\item Bridge effectiveness: in Cube\,\&\,Conquer, bridge budgets eliminate a constant
fraction of cubes at the split frontier (reported as per–frontier kill rates).
\end{itemize}

\paragraph{Potential failure modes and mitigations.}
\begin{itemize}[leftmargin=1.6em]
\item \emph{Depth decay.} $\hat\rho_d\!\to\!0$ beyond shallow depths (overlay too weak or
windows too correlated). \emph{Mitigation:} increase overlay girth/degree locally; retune
window sizes to reduce correlation.
\item \emph{Overhead inversion.} LP/verification dominates wall time. \emph{Mitigation:}
tighten cut cascade (pattern/library first), cache certificates by $(\alpha,\text{cut})$,
cap separation frequency.
\item \emph{Heuristic interference.} Cloning harms propagation/branching. \emph{Mitigation:}
clone only top-centrality variables; add clone-aware tie-breakers; keep channels two-watched.
\end{itemize}

\paragraph{Scope.}
The scenario above identifies a concrete path to Tier~C/D \emph{on round–reduced,
structured families} under measurable and falsifiable assumptions. It does not by itself
constitute a practical cryptanalytic break of standard-parameter primitives; persistence,
cost control, and independence must be demonstrated empirically as outlined.

\section{Policy and Standards Note (Short)}
\label{app:policy}
If future experiments demonstrate a non-negligible contextuality density $\rho_{\text{graph}}(G_\Phi)$ at scale, the
resulting $\delta>0$ would, in principle, translate into modest reductions of effective key lengths. This is a
theoretical possibility, not a present cryptanalytic capability. No immediate changes to practice are warranted;
the next step is empirical benchmarking of $\rho_{\text{graph}}$ in small clause networks.

\section{Figures}
\label{app:figs}

In this appendix we collect three figures that visualize (i) the base $2-\rho$
as a function of $\rho$, (ii) node counts with and without h\mbox{-}NGCC, and
(iii) per\mbox{-}depth pruning on $C_5$.

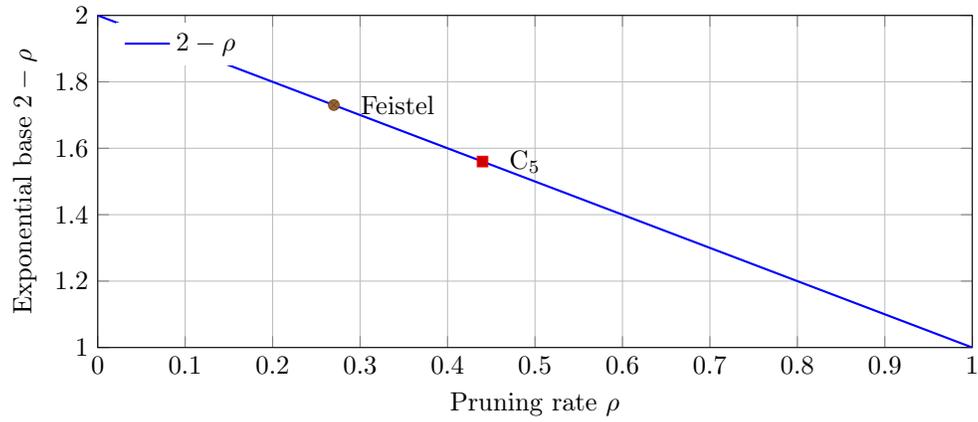
\begin{figure}[h]
  \centering
  \begin{tikzpicture}
    \begin{axis}[
      width=0.8\linewidth, height=6cm,
      xlabel={Pruning rate $\rho$}, ylabel={Exponential base $2-\rho$},
      xmin=0, xmax=1, ymin=1, ymax=2,
      grid=both, domain=0:1, samples=100,
      legend style={draw=none, at={(0.02,0.98)}, anchor=north west}
    ]
      \addplot+[mark=none, thick] {2 - x};
      \addlegendentry{$2-\rho$}

      % experimental points
      \addplot+[only marks] coordinates {(0.44,1.56)};
      \node[anchor=west] at (axis cs:0.44,1.56) {~~C$_5$};

      \addplot+[only marks] coordinates {(0.27,1.73)};
      \node[anchor=west] at (axis cs:0.27,1.73) {~~Feistel};
    \end{axis}
  \end{tikzpicture}
  \caption{Exponential base $2-\rho$ versus pruning rate $\rho$ with markers at the observed values from Appendix~\ref{app:empirical}.}
  \label{fig:base-vs-rho}
\end{figure}

\begin{figure}[h]
  \centering
  \begin{tikzpicture}
    \begin{axis}[
      ybar=0pt,
      width=0.8\linewidth, height=6cm,
      symbolic x coords={C$_5$,Feistel},
      xtick=data,
      ylabel={Nodes explored (avg)},
      ymin=0,
      grid=both,
      legend style={draw=none, at={(0.5,1.02)}, anchor=south}
    ]
      \addplot coordinates {(C$_5$,16) (Feistel,4096)};
      \addplot coordinates {(C$_5$,9) (Feistel,2800)};
      \legend{Vanilla, h-NGCC}
    \end{axis}
  \end{tikzpicture}
  \caption{Average nodes explored with and without h\mbox{-}NGCC pruning on toy instances.}
  \label{fig:nodes-bars}
\end{figure}
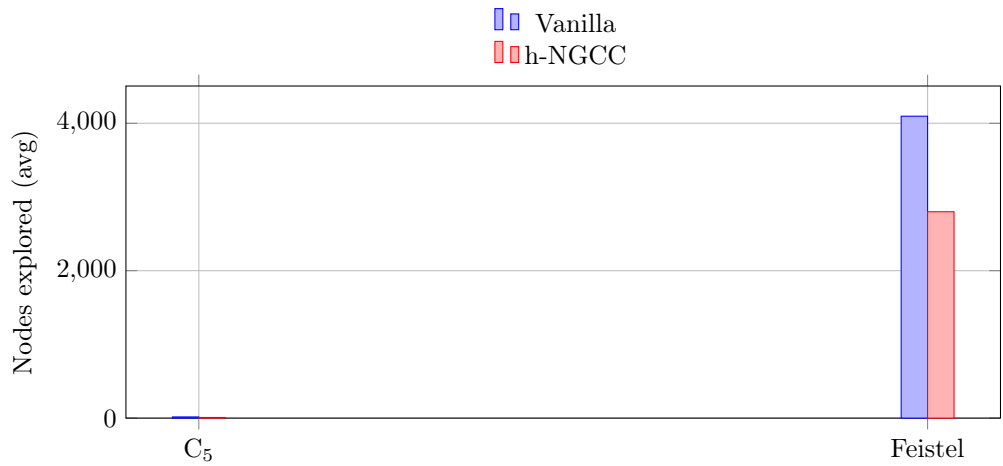

\begin{figure}[h]
  \centering
  \begin{tikzpicture}
    \begin{axis}[
      width=0.8\linewidth, height=6cm,
      xlabel={Depth $d$}, ylabel={$\hat{\rho}_d$},
      xmin=1, xmax=4, ymin=0, ymax=1,
      grid=both
    ]
      \addplot+[mark=*] coordinates {(1,0.5) (2,0.5) (3,0.5) (4,0.0)};
    \end{axis}
  \end{tikzpicture}
  \caption{Empirical pruning fraction by depth for the $C_5$ demo (illustrative).}
  \label{fig:rho-per-depth}
\end{figure}
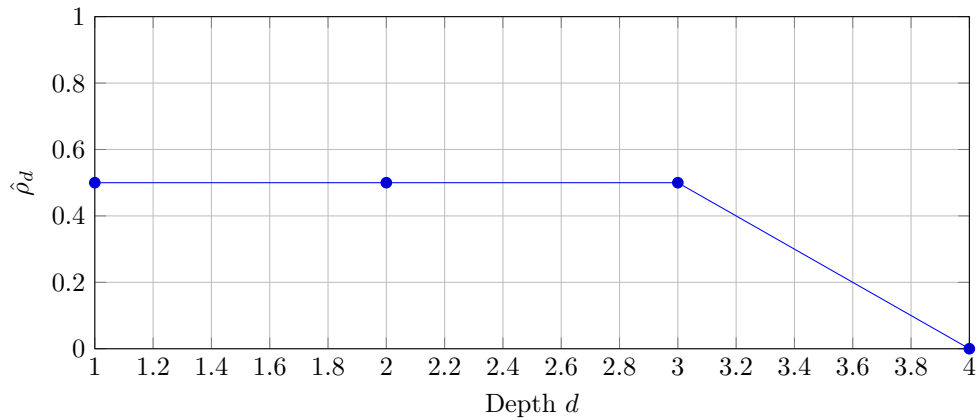

% C5 exclusivity (KCBS) graph
\begin{figure}[h]
  \centering
  \begin{tikzpicture}[baseline=(current bounding box.center)]
    \def\n{5}\def\r{2.0}
    % vertices on a regular 5-gon
    \foreach \i in {1,...,\n}{
      \node[vertex, labelfont] (v\i) at ({90-360/\n*(\i-1)}:\r) {$X_{\i}$};
    }
    % edges (contexts)
    \foreach \i [evaluate=\i as \j using {int(mod(\i,\n)+1)}] in {1,...,\n}{
      \draw[edge] (v\i) -- (v\j);
    }
  \end{tikzpicture}
  \caption{C$_5$ exclusivity graph (KCBS scenario): contexts are edges, variables are vertices.}
  \label{fig:c5}
\end{figure}
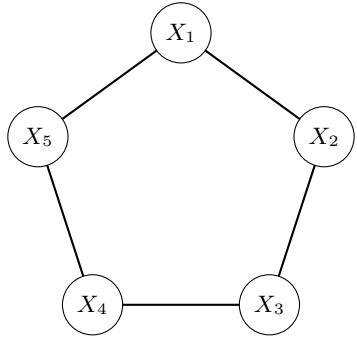

% C7 exclusivity graph
\begin{figure}[h]
  \centering
  \begin{tikzpicture}[baseline=(current bounding box.center)]
    \def\n{7}\def\r{2.2}
    % vertices on a regular 7-gon
    \foreach \i in {1,...,\n}{
      \node[vertex, labelfont] (w\i) at ({90-360/\n*(\i-1)}:\r) {$X_{\i}$};
    }
    % edges (contexts)
    \foreach \i [evaluate=\i as \j using {int(mod(\i,\n)+1)}] in {1,...,\n}{
      \draw[edge] (w\i) -- (w\j);
    }
  \end{tikzpicture}
  \caption{C$_7$ exclusivity graph: longer odd cycles also forbid global assignments and admit h-NGCC inequalities.}
  \label{fig:c7}
\end{figure}
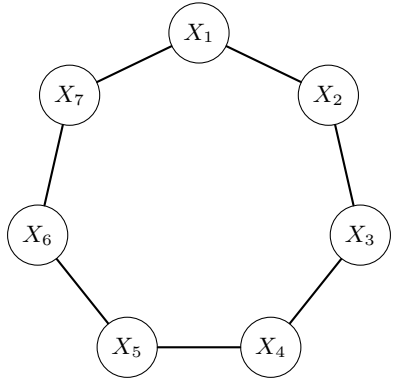

\begin{figure}[h]
  \centering
  \begin{tikzpicture}
    \begin{axis}[
      width=0.8\linewidth, height=6cm,
      xlabel={Visibility $\eta$}, ylabel={CHSH value $S$},
      xmin=0, xmax=1, ymin=0, ymax=3,
      grid=both, legend style={draw=none, at={(0.02,0.98)}, anchor=north west}
    ]
      % Classical bound
      \addplot+[domain=0:1, thick] {2};
      \addlegendentry{Classical bound $S=2$}

      % Quantum (Werner-state) line S = 2 sqrt(2) * eta
      \addplot+[domain=0:1, thick] {2*sqrt(2)*x};
      \addlegendentry{$S=2\sqrt{2}\,\eta$ (Tsirelson line)}
    \end{axis}
  \end{tikzpicture}
  \caption{CHSH visibility curve: quantum line $2\sqrt{2}\eta$ crosses the classical bound $S=2$ at $\eta=1/\sqrt{2}$.}
  \label{fig:chsh-visibility}
\end{figure}
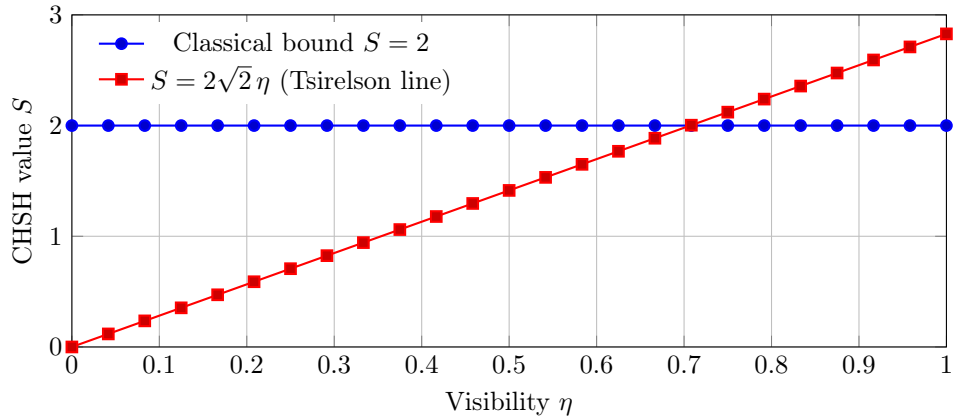

\begin{figure}[h]
  \centering
  \begin{tikzpicture}[node distance=18mm]
    % variables (top row)
    \node[vtx] (x1) {$x_1$};
    \node[vtx, right=22mm of x1] (x2) {$x_2$};
    \node[vtx, right=22mm of x2] (x3) {$x_3$};
    \node[vtx, right=22mm of x3] (x4) {$x_4$};

    % contexts (bottom row)
    \node[ctx, below=14mm of $(x1)!0.5!(x2)$] (c1) {$C_1$};
    \node[ctx, right=6mm of c1] (c2) {$C_2$};
    \node[ctx, below=14mm of $(x3)!0.5!(x4)$] (c3) {$C_3$};

    % tidy connectors (variables -> contexts they belong to)
    \draw[edge] (x1.south) -- ($(c1.north)+( -5mm,0)$);
    \draw[edge] (x2.south) -- ($(c1.north)+(  5mm,0)$);
    \draw[edge] (x2.south) -- ($(c2.north)+( -5mm,0)$);
    \draw[edge] (x3.south) -- ($(c2.north)+(  5mm,0)$);
    \draw[edge] (x3.south) -- ($(c3.north)+( -5mm,0)$);
    \draw[edge] (x4.south) -- ($(c3.north)+(  5mm,0)$);

    % overlap hint (optional): small bracket showing C1--C2 share x2
    % \draw[thick] ($(c1.north east)+(1mm,0)$) -- ($(c2.north west)+(-1mm,0)$);

    % subtitle
    \node[font=\small, align=center, below=4mm of c2]
      {Clauses as contexts; overlapping variables induce counterfactual channels};
  \end{tikzpicture}
  \caption{Schematic mapping of a CNF formula to an h\mbox{-}NGCC network: each clause becomes a context; shared variables create consistency edges.}
  \label{fig:clause-context-map}
\end{figure}

\begin{figure}[h]
  \centering
  \begin{tikzpicture}[node distance=28mm]
    % contexts
    \node[draw,rounded corners, minimum width=22mm, minimum height=10mm] (c1) {$C_1$};
    \node[draw,rounded corners, minimum width=22mm, minimum height=10mm, right=of c1] (c2) {$C_2$};
    \node[draw,rounded corners, minimum width=22mm, minimum height=10mm, right=of c2] (c3) {$C_3$};

    % overlaps
    \node[below=7mm of c1] (b1) {$\,(L_1,R_1)\,$};
    \node[below=7mm of c2] (b2) {$\,(L_2,R_2)\,$};
    \node[below=7mm of c3] (b3) {$\,(L_3,R_3)\,$};

    \draw[<->] (c1) -- (c2) node[midway, above] {overlap};
    \draw[<->] (c2) -- (c3) node[midway, above] {overlap};

    \draw (c1) -- (b1);
    \draw (c2) -- (b2);
    \draw (c3) -- (b3);

    \node[above=7mm of c1] {$k_1$};
    \node[above=7mm of c2] {$k_2$};
    \node[above=7mm of c3] {$k_3$};
  \end{tikzpicture}
  \caption{Three Feistel round contexts with shared boundary wires. Disagreement indicators $Y_1,Y_2,Y_3$ on these overlaps yield the cut $\mathbb{E}[Y_1+Y_2+Y_3]\le 2$.}
  \label{fig:feistel-hngcc}
\end{figure}
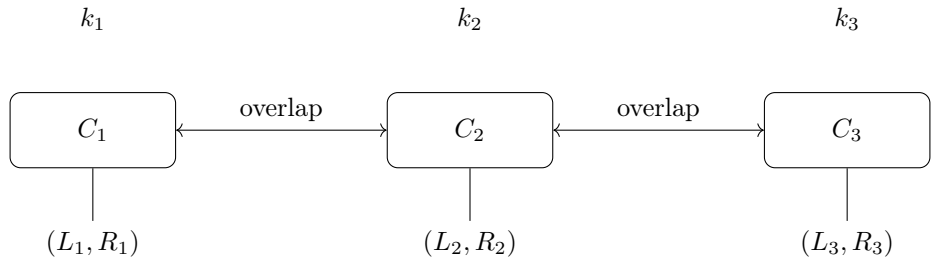

\begin{figure}[h]
  \centering
  \begin{tikzpicture}
    \begin{axis}[
      width=0.8\linewidth, height=6cm,
      xlabel={$V_L$}, ylabel={$V_R$},
      xmin=0, xmax=3, ymin=0, ymax=3, grid=both]
      % classical square V_L<=2, V_R<=2
      \addplot[fill=blue!10, draw=blue, thick] coordinates {(0,0) (2,0) (2,2) (0,2) (0,0)} -- cycle;
      % quantum budget line V_L+V_R = 4+sqrt(2)
      \addplot[red, thick, domain=0:3] {4 + sqrt(2) - x};
    \end{axis}
  \end{tikzpicture}
  \caption{Figure-of-8 budget: classical square $V_i\le 2$ and quantum line $V_L+V_R\le 4+\sqrt2$.}
  \label{fig:budget-fo8}
\end{figure}
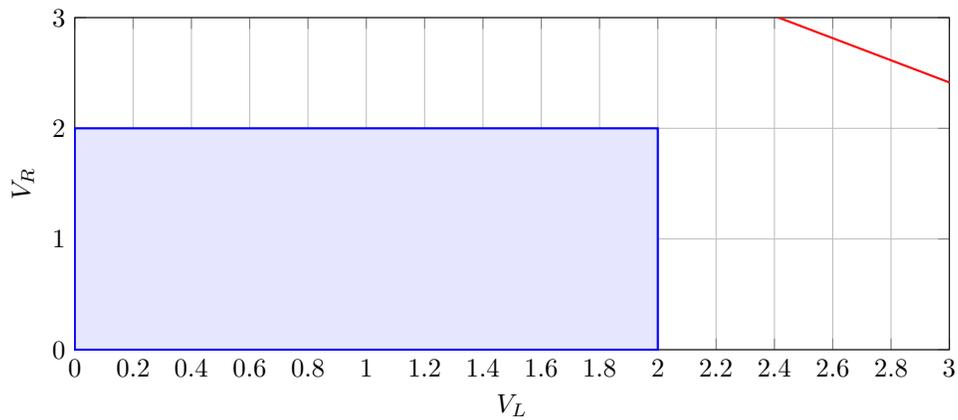

\begin{figure}[h]
  \centering
  \begin{tikzpicture}
    \begin{axis}[
      width=0.8\linewidth, height=6cm, view={35}{25},
      xlabel={$V_L$}, ylabel={$V_T$}, zlabel={$V_R$},
      xmin=0, xmax=3, ymin=0, ymax=3, zmin=0, zmax=3, grid=major]
      % classical cube (Vi<=2)
      \addplot3[surf, opacity=0.1, mesh/ordering=y varies]
        coordinates {(0,0,0) (2,0,0) (2,2,0) (0,2,0)
                     (0,0,2) (2,0,2) (2,2,2) (0,2,2)};
      % quantum plane V_L+V_T+V_R = 6+sqrt(2)
      \addplot3[surf, opacity=0.15, colormap name=redyellow]
        {6 + sqrt(2) - x - y};
    \end{axis}
  \end{tikzpicture}
  \caption{Cloverleaf budget: classical cube $V_i\le 2$ and quantum plane
  $V_L+V_T+V_R\le 6+\sqrt2$.}
  \label{fig:budget-clover}
\end{figure}
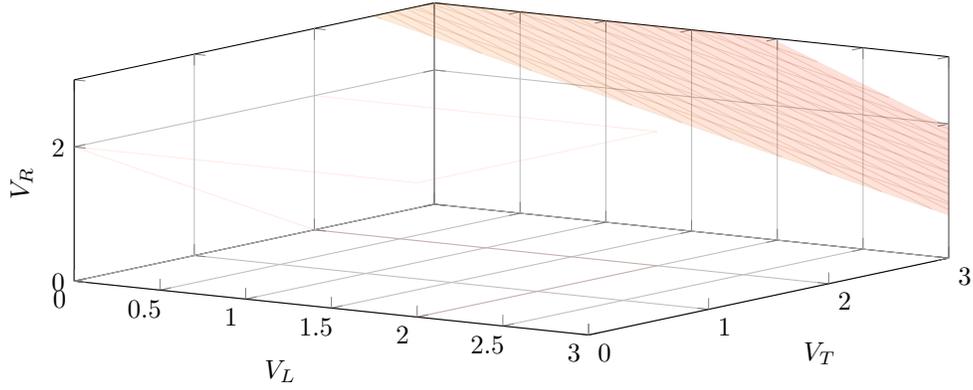

\begin{figure}[t]
\centering
% ---------------- (a) ----------------
\begin{tikzpicture}[>=Stealth, node distance=14mm, every node/.style={font=\small}]
\node[draw, circle, minimum size=6mm] (A) {A};
\node[draw, rectangle, minimum width=24mm, minimum height=10mm, left=of A]  (L) {Left 4-cycle};
\node[draw, rectangle, minimum width=24mm, minimum height=10mm, right=of A] (R) {Right 4-cycle};

\node[draw, fill=white, align=center, inner sep=3pt, above=12mm of A] (budget) {
  $\begin{aligned}
    V_L + V_R &\le 4 + \sqrt{2} &&\text{(quantum)}\\
    V_L + V_R &\le 4            &&\text{(classical)}
  \end{aligned}$
};

\draw[->] (L.east) -- (A.west);
\draw[->] (R.west) -- (A.east);

\node[below=8mm of A] (capA) {(a) Figure-of-8 overlaps (monogamy budget).};
\end{tikzpicture}

\vspace{12mm}

% ---------------- (b) ----------------
\begin{flushleft}
\begin{tikzpicture}[>=Stealth, node distance=14mm, every node/.style={font=\small}]
% Inputs
\foreach \i in {1,2,3} {
  \node[anchor=east] (l\i) at (0,-\i*7mm) {$\ell_{\i}$};
  \node[anchor=west, right=6mm of l\i] (phi\i) {$\phi_{\i}$};
  \node[draw, circle, minimum size=5pt, right=6mm of phi\i] (in\i) {};
  \draw (phi\i) -- (in\i);
}

% Tritter
\node[draw, rectangle, minimum width=30mm, minimum height=24mm,
      right=20mm of in2, align=center] (U)
{$\tfrac{1}{\sqrt{3}}
  \begin{pmatrix}
    1 & 1 & 1 \\
    1 & \omega & \omega^2 \\
    1 & \omega^2 & \omega
  \end{pmatrix}$};

% Outputs
\foreach \i/\name in {1/D1,2/D2,3/D3} {
  \path (U.east |- in\i) coordinate (u\i);
  \node[draw, rectangle, minimum width=12mm, right=18mm of u\i,
        fill={\ifnum\i=2 yellow!20\else none\fi}] (D\i) {\name\ifnum\i=2\ (witness $S$)\fi};
  \draw[->] (in\i) -- (U.west |- in\i);
  \draw[->] (U.east |- in\i) -- (D\i.west);
}

\node[below=8mm of U] {(b) Clause gadget with tritter witness (middle port is $S$).};
\end{tikzpicture}
\end{flushleft}

\caption{Monogamy “budget” and clause gadget with tritter witness.}
\label{fig:fig8-clause-tritter}
\end{figure}   % <--- CLOSE the figure environment properly here

\clearpage % or \FloatBarrier if you use \usepackage{placeins}

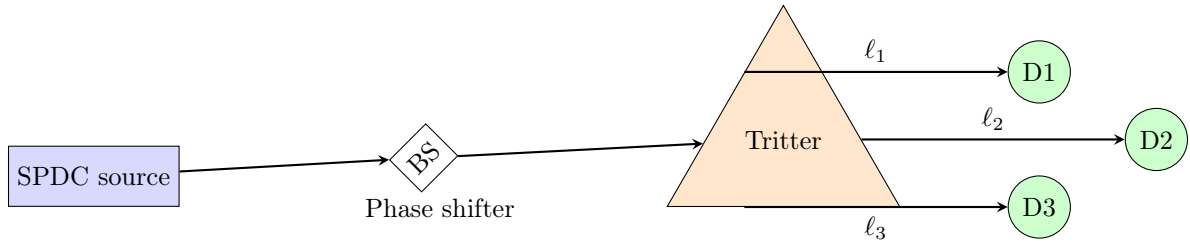
\begin{figure}[H]
\centering
\begin{tikzpicture}[scale=1.0,>=stealth]

% Photon source
\node[draw, rectangle, fill=blue!15, minimum width=1.8cm, minimum height=0.8cm] (src) {SPDC source};

% Beam splitter
\node[draw, rectangle, rotate=45, minimum size=0.6cm, right=3cm of src] (bs1) {BS};

% Tritter
\node[draw, regular polygon, regular polygon sides=3, minimum size=1.5cm, fill=orange!20, right=3.5cm of bs1] (trit) {Tritter};

% Detectors (stacked vertically)
\node[draw, circle, fill=green!20, right=3.5cm of trit.120] (d1) {D1};
\node[draw, circle, fill=green!20, right=3.5cm of trit] (d2) {D2};
\node[draw, circle, fill=green!20, right=3.5cm of trit.240] (d3) {D3};

% Paths
\draw[->, thick] (src) -- (bs1);
\draw[->, thick] (bs1) -- (trit);

\draw[->, thick] (trit.120) -- (d1.west) node[midway, above] {$\ell_1$};
\draw[->, thick] (trit) -- (d2.west) node[midway, above] {$\ell_2$};
\draw[->, thick] (trit.240) -- (d3.west) node[midway, below] {$\ell_3$};

% Label
\node[below=0.2cm of bs1] {Phase shifter};

\end{tikzpicture}
\caption{Minimal experimental setup. A heralded photon from an SPDC source passes
through a beam splitter and phase shifter, then enters a tritter implementing the clause gadget.
Three detectors (D1–D3) register outcomes. All components are standard laboratory devices,
with no exotic or black-box hardware required.}
\label{fig:exp-setup}
\end{figure}

\begin{figure}[H]
\centering
\resizebox{0.85\textwidth}{!}{%
\begin{tikzpicture}[>=stealth]

% Photon source
\node[draw, rectangle, fill=blue!15, minimum width=1.6cm, minimum height=0.7cm] (src) {SPDC};

% First beam splitter
\node[draw, rectangle, rotate=45, minimum size=0.5cm, right=2.2cm of src] (bs1) {BS};

% Upper Mach–Zehnder arm
\node[draw, rectangle, rotate=45, minimum size=0.5cm, above right=1.3cm and 2.0cm of bs1] (bs2u) {BS};
\node[draw, rectangle, rotate=45, minimum size=0.5cm, right=2.0cm of bs2u] (bs3u) {BS};

% Lower Mach–Zehnder arm
\node[draw, rectangle, rotate=45, minimum size=0.5cm, below right=1.3cm and 2.0cm of bs1] (bs2l) {BS};
\node[draw, rectangle, rotate=45, minimum size=0.5cm, right=2.0cm of bs2l] (bs3l) {BS};

% Detectors
\node[draw, circle, fill=green!20, right=2.0cm of bs3u] (d1) {D1};
\node[draw, circle, fill=green!20, right=2.0cm of bs3l] (d2) {D2};

% Paths
\draw[->, thick] (src) -- (bs1);

% Upper path
\draw[->, thick] (bs1.45) -- ++(1.4,1.4) -- (bs2u.west);
\draw[->, thick] (bs2u.45) -- ++(1.6,0) -- (bs3u.west);
\draw[->, thick] (bs3u.45) -- (d1.west);

% Lower path
\draw[->, thick] (bs1.315) -- ++(1.4,-1.4) -- (bs2l.west);
\draw[->, thick] (bs2l.315) -- ++(1.6,0) -- (bs3l.west);
\draw[->, thick] (bs3l.315) -- (d2.west);

% Labels
\node[below=0.2cm of bs1] {Phase shifter};
\node[above=0.2cm of d1] {$\ell_{\text{upper}}$};
\node[below=0.2cm of d2] {$\ell_{\text{lower}}$};

\end{tikzpicture}%
}
\caption{Compact nested Mach–Zehnder interferometer. A heralded photon
from an SPDC source enters cascaded interferometers with phase shifters
in the arms. Detection at D1 or D2 provides statistics for
contextuality-type inequalities. All components are standard optics.}
\label{fig:nested-mzi-compact}
\end{figure}
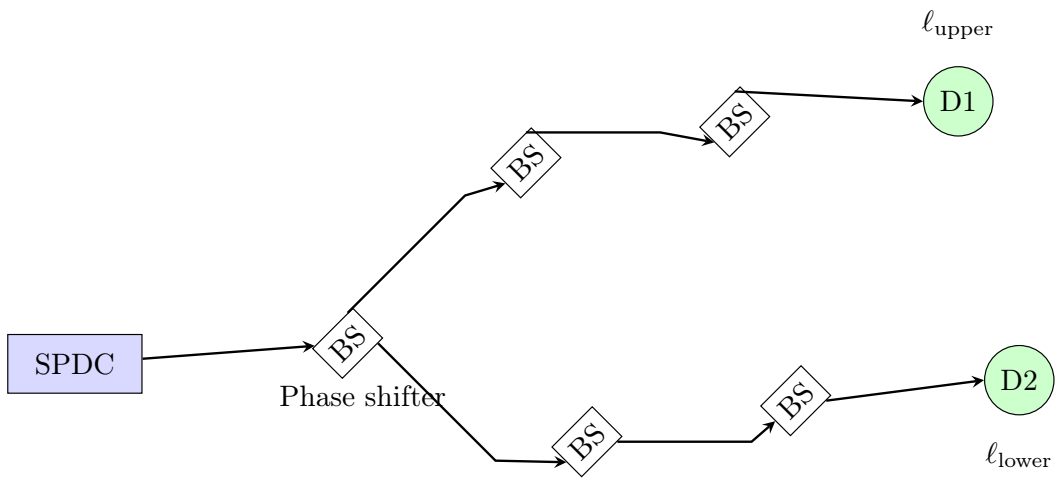

\clearpage      % flush floats before references
\bibliographystyle{unsrt}   % only one style
\bibliography{refs}         % match your refs.bib file

\begin{thebibliography}{1}

\bibitem{vonLiechtensteinCLF2025}
Maximilian Ralph~Peter von Liechtenstein.
\newblock Counterfactual logic framework for $\varepsilon$-instrumentation, 2025.

\bibitem{AbramskyBrandenburger2011}
Samson Abramsky and Adam Brandenburger.
\newblock The sheaf-theoretic structure of non-locality and contextuality.
\newblock {\em New Journal of Physics}, 13:113036, 2011.

\bibitem{Laurent2003}
Monique Laurent.
\newblock A comparison of sherali--adams, lov{\'a}sz--schrijver, and lasserre relaxations for 0--1 programming.
\newblock {\em Mathematics of Operations Research}, 28(3):470--496, 2003.

\bibitem{Biere2009}
Armin Biere, Marijn Heule, Hans van Maaren, and Toby Walsh, editors.
\newblock {\em Handbook of Satisfiability}, volume 185 of {\em Frontiers in Artificial Intelligence and Applications}.
\newblock IOS Press, 2009.

\bibitem{Lovasz1979}
L\'aszl\'o Lov\'asz.
\newblock On the shannon capacity of a graph.
\newblock {\em IEEE Transactions on Information Theory}, 25(1):1--7, 1979.

\bibitem{Juhasz1983}
Ferenc Juh\'asz.
\newblock On the asymptotic behaviour of lov\'asz' theta function for random graphs.
\newblock {\em Discrete Mathematics}, 46:161--172, 1983.

\end{thebibliography}

\end{document}